\documentclass[acmsmall]{acmart}

\usepackage{balance}  
\usepackage{times}
\usepackage{graphicx}
\usepackage{epstopdf}
\usepackage{latexsym}
\usepackage{algorithm}
\usepackage[noend]{algpseudocode}
\usepackage{epsfig}
\usepackage{amsfonts}
\usepackage{amsmath}
\usepackage{amsthm}
\usepackage{subfigure}
\usepackage{stmaryrd}
\usepackage{bm}
\usepackage{algpseudocode}
\usepackage{latexsym}
\usepackage{setspace}
\usepackage{varwidth}
\usepackage{tabu}
\usepackage{stfloats}
\usepackage{multirow}
\usepackage[english]{babel}
\usepackage{url}
\usepackage{xcolor}
\usepackage{pdflscape}
\usepackage{graphicx}
\usepackage{float}
\usepackage{comment}
\usepackage{ulem}
\usepackage{multibib}

\newtheorem{Def}{\textbf{Definition}}
\newtheorem{Lem}{\textbf{Lemma}}
\newtheorem{Cor}{\textbf{Corollary}}
\newtheorem{example}{\textbf{Example}}
\newtheorem{rul}{\textbf{Rule}}

\AtBeginDocument{%
  \providecommand\BibTeX{{%
    \normalfont B\kern-0.5em{\scshape i\kern-0.25em b}\kern-0.8em\TeX}}}

\setcopyright{acmcopyright}
\copyrightyear{2022}
\acmYear{2022}
\acmDOI{10.1145/1122445.1122456}

\acmBooktitle{Woodstock '18: ACM Symposium on Neural Gaze Detection,
  June 03--05, 2018, Woodstock, NY}
\acmPrice{15.00}
\acmISBN{978-1-4503-XXXX-X/18/06}

\begin{document}

\title{Online Discovery of Evolving Groups over Massive-Scale Trajectory Streams}

\author{Yanwei Yu}
\authornote{This is the corresponding author.}
\email{yuyanwei@ouc.edu.cn}
\author{Ruoshan Lan}
\email{lanruoshan@gmail.com}
\affiliation{%
  \institution{Ocean University of China}
  \streetaddress{238 Songling RD}
  \city{Qingdao}
  \state{Shandong}
  \country{China}
  \postcode{266100}
}

\author{Lei~Cao}
\affiliation{%
  \institution{Massachusetts Institute of Technology}
  \city{Cambridge}
  \state{MA}
  \country{USA}
  \postcode{02139}}
\email{lcao@csail.mit.edu}

\author{Peng~Song}
\email{pengsong@ytu.edu.cn}
\author{Yingjie Wang}
\email{wangyingjie@ytu.edu.cn}
\affiliation{%
  \institution{Yantai University}
  \streetaddress{30 Qingquan RD}
  \city{Yantai}
  \state{Shandong}
  \country{China}
  \postcode{264005}
}

\renewcommand{\shortauthors}{Yu and Lan, et al.}

\begin{abstract}
The increasing pervasiveness of object tracking technologies leads to huge volumes of spatiotemporal data collected in the form of trajectory streams.
The discovery of useful group patterns from moving objects' movement behaviours in trajectory streams is critical for real-time applications ranging from transportation management to military surveillance.
Motivated by this, we first propose a novel pattern, called evolving group, which models the unusual group events of moving objects that travel together within density connected clusters in evolving streaming trajectories. Our theoretical analysis and empirical study on the Osaka Pedestrian data and Beijing Taxi data demonstrate its effectiveness in capturing the development, evolution, and trend of group events of moving objects in streaming context.
Moreover, we propose a discovery method that efficiently supports online detection of evolving groups over massive-scale trajectory streams using a sliding window. It contains three phases along with a set of novel optimization techniques designed to minimize the computation costs.
Furthermore, to scale to huge workloads over evolving streams, we extend our discovery method to a parallel framework by using a sector-based partition.
Our comprehensive empirical study demonstrates that our online discovery framework is effective and efficient on real-world high-volume trajectory streams. 
\end{abstract}

\ccsdesc[500]{Information systems~Spatial-temporal systems}
\ccsdesc[500]{Information systems~Data stream mining}
\ccsdesc[500]{Computing methodologies~Parallel algorithms}

\keywords{Moving objects, pattern mining, evolving group pattern, trajectory streams}

\maketitle

\section{Introduction}
\label{sec.intro}

In recent years, location tracking technologies have been broadly utilized in a variety of applications ranging from traffic management to mobile social networks, which have generated huge volumes of trajectory data from moving objects including people, vehicles and animals, etc.  
Such trajectory data can be utilized for different purposes, such as travel-route prediction, friends recommendation, anomaly detection, and traffic control~\cite{zheng2015trajectory}\cite{li2011movemine}\cite{yu2014detecting}.  In this work, we focus on detecting a particular type of movement pattern called \textit{evolving group} from massive-scale moving object trajectory streams effectively and efficiently.

Evolving group pattern can be considered as a special type of group pattern that models the behavior of the moving objects that travel together over time. Techniques have been proposed in the literature to detect group patterns such as $flock$~\cite{benkert2008reporting}, $convoy$~\cite{jeung2008convoy}\cite{jeung2008discovery}, $swarm$~\cite{li2010swarm} and $gathering$~\cite{zheng2013discovery}\cite{zheng2014online}.
A flock \cite{benkert2008reporting} consists of objects that travel together within a user-specified distance range threshold. A convoy pattern in \cite{jeung2008convoy}\cite{jeung2008discovery} is defined as a set of objects that move together (always falling into the same density-based clusters) during at least $k$ consecutive timestamps. Swarm pattern \cite{li2010swarm} is a variation of convoy pattern. It allows the moving objects to leave the swarm temporally. However, their objective is to discover groups of objects that move together in static trajectory database.
Among these techniques, the $gathering$ pattern is closest to our evolving group pattern.  
A gathering is defined as a sequence of density-based clusters for a period of at least $k_c$ consecutive timestamps in which adjacent clusters are within a given distance range $d$. Moreover, each cluster is required to contain at least $m_p$ dedicated members (so-called $participators$) who appear in at least $k_p$ clusters, although not necessarily to be consecutive.
Similar to the swarm pattern from \cite{li2010swarm}, members in a gathering group also enter and leave the group.
However, the $gathering$ pattern requires at least $k_c$ consecutive timestamps, which might result in the loss of interesting sequence of moving object clusters. Furthermore, understanding the trend and evolution of group events in streaming environments is considered more useful and helpful than simply extracting group patterns along time. For example, the causal interactions of discovered groups in time sequence can reveal the inherent relationships of groups of moving objects.

Figure~\ref{example} illustrates an example of gathering patterns. Let $m_c$=$3$ (minimal number of objects for a cluster), $k_c$=$3$ (minimal duration of a gathering), $k_p$=$2$ (minimal number of participated clusters for a participator) and $m_p$=$3$ (minimal number of participators). Suppose cluster $c_3$ is too far away from $c_2$, and $c_4$ is also far away from $c_1$, namely, $distance(c_2,c_3)>d$ and $distance(c_1,c_4)>d$. From this example, two sequences of clusters $\langle c_1,c_3,c_5\rangle$ and $\langle c_2,c_4,c_5\rangle$ form two gathering candidates from $t_1$ to $t_3$. In this case, only $\langle c_1,c_3,c_5\rangle$ is a gathering pattern with participator set $\{o_1,o_2,o_3,o_4\}$ because the sequence contains at least three participators at each cluster, whereas $\langle c_2,c_4,c_5\rangle$ (participator set $\{o_3,o_4,o_5,o_6\}$) only contains two participator $\{o_3,o_4\}$ in $c_5$. Similarly, from $t_5$ to $t_7$, there are two gathering candidates $\langle c_6,c_7,c_9\rangle$ and $\langle c_6,c_8,c_9\rangle$. However, only $\langle c_6,c_7,c_9\rangle$ forms a gathering since each cluster includes at least three participators.

\begin{figure}[htbp]
\centering
\includegraphics[width=0.85\textwidth]{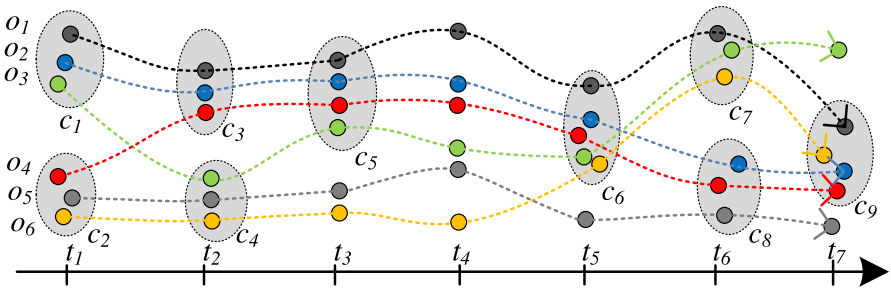}
\caption{Comparison of gathering and evolving group. A small circle represents a trajectory point. A large gray dashed oval denotes a cluster.}
\label{example}
\end{figure}

Using \textit{gathering} pattern, we can get two independent gatherings from this example. This is because gathering requires the group to move in consecutive clusters during its lifetime. However, \textit{do the two gatherings have relationship with each other?} We find that there are four participators $\{o_1,o_2,o_3,o_4\}$ in the gathering $\langle c_1,c_3,c_5\rangle$ and five participators $\{o_1,o_2,o_3,o_4,o_6\}$ in $\langle c_6,c_7,c_9\rangle$. 
The four participators in the first gathering $\langle c_1,c_3,c_5\rangle$ still take an active part in the second gathering $\langle c_6,c_7,c_9\rangle$ and are also detected as participators. 
Additionally, these two gatherings are close to each other in both aspects of location and time. 
Hence gathering $\langle c_6,c_7,c_9\rangle$ is actually developed from $\langle c_1,c_3,c_5\rangle$ and grows larger. 
Therefore, the $gathering$ pattern is not suitable for modeling the variational groups in which memberships evolve gradually in streaming environments. 

To overcome the aforementioned shortcomings of gathering pattern, we try to define an evolving group as a dense group of objects that share common behaviors in most of the time and change gradually over time. 
Although gathering pattern is regarded as a dense and continuing group of individuals, we regard that it is unrealistic to use such a strict criterion of consecutive moving object clusters to track the practical group events.
For example, the crowding might become ``eased somewhat'' at some time points in a traffic jam, which results in there is no enough dense clusters at that time. But they might again fall into a congestion soon at next crossroad. The lightweight timestamp might split the whole group of objects ponderously. 
The another important characteristic of evolving group is that there is no requirement for coherent membership of group over streams, namely, the members could change gradually over time. Gathering pattern also allows members to enter and leave group at any time, however, gathering does desire the stable $participator$ to participate the group event during the whole lifetime period, while evolving group should capture the gradual change of \textit{participator} in streams, i.e., the evolution of participators.
Furthermore, the definition of evolving group should be suitable for capturing the group patterns from high-volume evolving streaming data. In \cite{zheng2014online}, a solution is proposed to support online discovery of gatherings in incremental manner, however, it does not take the moving object evolving trajectories into account, e.g., monitoring the dynamic data of the most recent five minutes over time.

In this article, we propose a new group pattern definition of \textit{evolving group} by introducing
the notions of ``crowd'' and ``aggregation'' \textbf{in sliding window} to capture the evolving aggregation of objects over trajectory streams.
Specifically, we define our `crowd'' to model the congregation of moving objects with an relaxed timestamp constrain in a sliding window.
The ``aggregation'' is a special ``crowd'' that contains at least $m_p$ participators in each cluster. Then we define the ``group'' as the set of participators in the ``aggregation'' of a window to model the \textit{core members}. We further define the \textit{evolving group} to refer a sequence chain of evolved ``groups'' that share most core members (e.g., $m_g$) in adjacent windows to capture the evolution of \textit{core members}.

In Figure \ref{example}, let window size be 4, $m_g$=$3$, and other parameters are same with $gathering$ pattern. There is a ``aggregation'' $\langle c_1,c_3,c_5\rangle$ along with a ``group'' $g_1$=$\{o_1,o_2,o_3,o_4\}$ in window $W_{t_1}^{t_4}$. Suppose $distance(c_5,c_6)<d$, $\langle c_3,c_5,c_6\rangle$ also forms a ``aggregation'' in $W_{t_2}^{t_5}$ and evolves from $\langle c_1,c_3,c_5\rangle$ with same ``group'' $g_2$=$\{o_1,o_2,o_3,o_4\}$. After the time $t_6$ arrives, the window slides from $W_{t_2}^{t_5}$ to $W_{t_3}^{t_6}$. Two ``crowds'' $\langle c_5,c_6,c_7\rangle$ and $\langle c_5,c_6,c_8\rangle$ are discovered, whereas only $\langle c_5,c_6,c_7\rangle$ is a ``aggregation'' that corresponds to the ``group'' $g_3$=$\{o_1,o_2,o_3,o_4,o_6\}$, which evolves from $g_2$ and grows larger. In window $W_{t_4}^{t_7}$, there is also only one ``aggregation'' $\langle c_6,c_7,c_9\rangle$ that satisfies $m_p$ participators at each cluster, and the corresponding ``group'' $g_4$=$\{o_1,o_2,o_3,o_4,o_6\}$ evolves and keeps from last window. Therefore, an \textit{evolving group} $\langle g_1, g_2, g_3, g_4\rangle$ reveals the development and trend of group events shown in the example.

The second challenge we must deal with is the \textit{design of a solution that can efficiently discover evolving groups from massive-scale streaming trajectories.}  Many applications require the real time monitoring and tracking of moving objects to discover the groups as soon as possible for supporting real time decision-making. Hence the discovering algorithm should report the results simultaneously while receiving and processing the massive-scale trajectory streams.
Apparently we can not simply apply or extend existing mining algorithms of group patterns to support online discovery of evolving groups.
In this article, we first propose an efficient discovery framework of the \textit{evolving group}, which contains three phases: 1) discover all closed ``crowds'' in the current window, 2) detect all closed ``aggregations'' with their corresponding ``groups'' from the discovered crowds, and 3) update the evolving groups. Furthermore, we design a multi-threading based parallel solution to online discover the evolving groups over high-volume trajectory streams, called \underline{m}ulti-\underline{t}hreading based \underline{o}nline \underline{d}iscovery (MTOD) framework.

We conduct extensive experiments on three real world datasets (\textit{Pedestrian}~\cite{zanlungo2014potential}\cite{zanlungo2015spatial}, \textit{Taxi}~\cite{yuan2013t}, and citywide surveillance \textit{Traffic} data in Jinan, China) to demonstrate the effectiveness and efficiency of our framework.

\textbf{Contributions.} The main contributions of this work are as follow:
(1) We propose a novel concept of evolving group that captures the interesting group events of moving objects in the stream context.
(2) Our empirical study using Pedestrian and Taxi datasets demonstrates the effectiveness of our new proposed group definition in capturing evolving groups over trajectory streams.
(3) We design an efficient evolving group discovery method that efficiently discovers the evolving groups from streaming trajectories in an incremental manner.
(4) We further enhance the proposed discovery method by proposing a multi-threading based parallel framework to online discover evolving groups over massive-scale moving object trajectory streams.
(5) Extensive evaluation study is conducted on the real large-scale traffic datasets to evaluate the efficiency of the proposed framework in near real time.

\textbf{Extension from the Conference Version.} While this work is based on a conference
article~\cite{lan2017discovering}, the scope of the proposed work has been significantly extended.

\begin{itemize}
\item We now extend our previous algorithm to a scalable online discovery framework based on multi-threading for evolving group mining over massive-scale trajectory streams (Sec.~\ref{sec.mtod}).
\item We now give a formal complexity analysis of our proposed algorithms (Sec.~\ref{subsec.complex}). In addition, the complexity analysis is also validated by the results of our extensive efficiency evaluation. 
Moreover, we also elaborate the Incre and VRA algorithms in this article.
\item We significantly extend the experimental evaluation, focusing on the effectiveness of our proposed evolving group definition (Sec.~\ref{subsec.effect}). Specifically, we now evaluate the effectiveness of the proposed group definition on a new real world dataset that contains the group annotations~\cite{zanlungo2014potential}\cite{zanlungo2015spatial} compared against state-of-the-art \textit{gathering} pattern~\cite{zheng2014online}.
\item We conduct an additional experimental study to demonstrate the efficiency of the newly proposed parallel framework based on multi-threading on one additional large-scale real-world dataset (the \textit{Traffic} data) in Sec.~\ref{subsec.scale}. 
\end{itemize}

The remainder of this article is organized as follows. We discuss related work which is related to our model in next section. We define the necessary concepts and formulate the focal problem of this paper in Section~\ref{sec.def}. Efficient methods for discovering evolving groups on archived and new arrivals of trajectory data are presented in Section~\ref{sec.alg}. Online discovery solution based on multi-threading techniques is presented in Section~\ref{sec.mtod}. Section~\ref{sec.exp} reports the experimental observations. 
Section~\ref{sec.con} concludes the article.

\section{Related Work}
\label{sec.relwork}
In this section, we mainly review the representative work that are related to our problem in the three areas of trajectory cluttering, group pattern mining in static trajectory databases, and in streaming trajectories.

\textbf{Trajectory clustering.}
Gaffney et al.~\cite{gaffney1999trajectory} first propose the fundamental principles of clustering trajectories based on probabilistic modeling. They consider the trajectory as a whole and represent a set of trajectories using a mixture of regression models. Then an unsupervised learning is carried out using the maximum likelihood principle. In particular, they use EM algorithm to estimate hidden parameters involved in probability models, and then determine the clusters membership. 
Fr\'echet distance is proposed to measure the similarity between curves~\cite{eiter1994computing}\cite{alt1995computing}. In the recent years, it is widely used to measure the similarities between trajectories~\cite{toohey2015trajectory}\cite{xie2017distributed}. 
More recently, Jin et  al.~\cite{Jin2019MovingOL} study the problem of moving object linking based on their historical traces. They define moving object linking as a $k$-nearest neighbour search problem on the collection of signatures and aim at measuring the similarity of two trajectories.
As pointed out by Lee et al.~\cite{lee2007trajectory}, distance measure based on whole trajectories may miss interesting common paths in sub-trajectories.
Lee et al. \cite{lee2007trajectory} propose a partition and group framework to discover common sub-trajectory clusters in static trajectory databases. They first partition trajectories into a set of quasi-linear segments using MDL principle. Then the line segments are grouped using density-based clustering to find common sub-trajectory clusters.
Lee et al.~\cite{lee2008traclass} succeedingly propose a hierarchical feature generation framework for trajectory classification by partitioning trajectories and exploring region-based and trajectory-based clustering.
Jensen et al.~\cite{jensen2007continuous} present a scheme that is capable of incrementally clustering moving objects by employing object dissimilarity and clustering features to improve the clustering effectiveness. 
Li et al.~\cite{li2010incremental} further propose an incremental trajectory clustering framework that contains online micro-cluster maintenance and offline macro-cluster creation for incremental trajectory databases.
\textit{Unlike the group pattern mining focused in this paper, this category of proposals only regard trajectories as sequences of line segments without considering temporal information, resulting in objects whose trajectory points fall in the same clusters may not actually move together in time domain.} 

\textbf{Co-locating pattern discovery in static databases.}
One of the earliest works on co-locating pattern discovery is introduced by Laube and Imfeld~\cite{laube2002analyzing} and further $flock$~\cite{benkert2008reporting}\cite{vieira2009line}, $convoy$~\cite{jeung2008convoy}, $swarm$~\cite{li2010swarm} and $gathering$~\cite{zheng2013discovery} are studied by others.
Kalnis et al.~\cite{kalnis2005discovering} propose one similar notion of moving cluster, which is a set of objects when they can be clustered at consecutive time points, and the portion of common objects in any two consecutive clusters is not below a predefined threshold.
Another similar notion, \textit{moving group pattern}~\cite{wang2006efficient}, relies on disk-based clustering to mine a set of objects that are within a distance threshold from one another for a minimum duration. Similar with $swarm$, \textit{moving group pattern} also permits members of pattern to travel together for a number of nonconsecutive timestamps by a weight threshold.
A recent study by Li et al.~\cite{li2013effective} proposes the notion of $group$ that uses density connectedness for clustering trajectories without relying on sampling points. Namely, the group pattern simultaneously satisfies sampling independence, density connectedness, trajectory approximation and online processing.
Recently, a kind of loose group movement pattern~\cite{wang2015discovering}\cite{naserian2016discovery}\cite{naserian2018framework} is proposed to discover the groups with coherent members and strict density connectedness constraint among members from trajectories. 
\textit{Such work is to discover the clusters of coherent objects that move together in static trajectory databases, which are quit different from our evolving group mining in the context of trajectory streams.} 

\textbf{Group pattern discovery in trajectory streams.}
Tang et al~\cite{tang2012discovery} recently propose the problem of discovering travelling companions in the context of streaming trajectories. The notion of travelling companion is as essential as $convoy$~\cite{jeung2008convoy}. However, they work on incremental algorithm of pattern discovery when the trajectories of users arrive in form of data streams.
Vieira et al.~\cite{vieira2009line} present the online flock discovery solution in polynomial time by identifying a discrete number of locations to place the center of the flock disk inside the spatial universe. They further propose a framework that uses a lightweight grid-based structure to efficiently and incrementally process the trajectory locations to discover flock patterns in streaming spatiotemporal data.
Aung and Tan~\cite{aung2010discovery} propose the notion of evolving convoys to better understand the states of convoys. Specifically, an evolving convoy contains both dynamic members and persisted members. As time passes, the dynamic members are allowed to move into or out of the evolving convoy, creating many stages of the same convoy. At the end, the evolving convoys with their stages are returned.
Yu et al.~\cite{yu2013online} propose an online clustering over trajectory streams to discover the groups with coherent members during consecutive timestamps. They first perform density-based clustering on trajectory line segments, and then update trajectory clusters inclemently.
Zheng et al.~propose an online discovery of ~\textit{gathering pattern}~\cite{zheng2013discovery} over trajectories in an incremental manner in \cite{zheng2014online}, which can capture congregations of moving individuals incrementally from durable and stable area with high density in trajectories.
\textit{However, our goal is totally different from these work. These studies try to find object clusters for consecutive duration of time points over trajectories, while our work attempts to discover evolving groups of dynamic core objects in most recent trajectory streams.}

\begin{table}[h]
\begin{center}
\caption{Notations and definitions}
\label{notations}
\begin{tabular}{p{1.3cm}|p{6.7cm}}
\hline
Notation  & Definition \\
\hline \hline
$o_{i}$ & a moving object   \\
\hline
$O_{DB}$ & the set of all moving objects   \\
\hline
$t_{i}$ & timestamp at the $i^{th}$ time point   \\
\hline
$c_{t_i}^j$ & a cluster at timestamp $t_i$ \\
\hline
$C_{t_i}$ & the collection of clusters at timestamp $t_i$ \\
\hline
$Cr$ & a crowd in a sliding window  \\
\hline
$k_c$ & the timestamp threshold of a crowd  \\
\hline
$m_c$ & the support threshold of a crowd  \\
\hline
$d$ & the distance threshold in crowd \\
\hline
$k_{p}$ & the cluster support threshold of a participator  \\
\hline
$m_{p}$ & the participator count threshold of an aggregation  \\
\hline
$m_{g}$ & the support threshold for evolving group \\
\hline
$k_{g}$ & the lifetime support threshold for evolving group \\
\hline
$W$, $w$ & a sliding window and its window size \\
\hline
$W_s$,$W_e$ & the starting and ending time of a window $W$  \\
\hline
$Ag$ & an aggregation in a window  \\
\hline
$Gr$ & a group in a window \\
\hline
$CanSet$ & the set of closed crowd candidate in a window \\
\hline
$endclu$ & the set of ending clusters of $CanSet$ \\
\hline
$c_{t_i}.st$ & the status flag of ending cluster $c_{t_i}$ \\
\hline
$o_{i}.cnt$ & the \# of clusters where $o$ appears in a crowd \\
\hline
\end{tabular}
\end{center}
\end{table}

\section{Problem Definition}
\label{sec.def}
In this section, we first introduce the definitions of all concepts used throughout the paper, and then formally state the focal problem to be solved. The list of major symbols and notations in this article is summarized in Table~\ref{notations}. 

Let $O_{DB}$  = $\{o_{1},o_{2},\dots,o_{n}\}$ be the set of all moving objects and $t_{i}$ be the timestamp at the $i^{th}$ time points in the trajectory streams.
We adopt the notion of density-based clustering \cite{ester1996density} to define the snapshot cluster.
A snapshot cluster is a group of objects with arbitrary shape and size, which are density-connected to each other at a given timestamp.
$C_{t_i}$=$\{c_{t_i}^1,c_{t_i}^2,\dots,c_{t_i}^m\}$ is the collection of snapshot clusters at timestamp $t_i$. 
Notice that we suppose that each snapshot is a short time interval, and the trajectory points of all moving objects in the same snapshot are received at the same timestamp, that is, the trajectory points of all moving objects generated within this interval are considered to be synchronized.  

In this work, we use the periodic sliding window semantics to define the sub-stream of an infinite trajectory data stream.
Each window $W$ has a starting time $W.s$ and an ending time $W.e=W.s+w-1$, where $w$ is a predefined window size. The window whose $W.e$ equals to the current timestamp is called the current window, denoted as $W_c$. We also use $W_{t_i}^{t_j}$ to denote the window whose $W.s=t_i$ and $W.e=t_j$.
Periodically the current window slides, causing $W.s$ and $W.e$ to increase by one timestamp.
Now we first introduce our ``crowd'' concept in sliding window as following Definition \ref{Def-crowd}.

\begin{Def}[Crowd]
\label{Def-crowd}
Given a trajectory stream in the sliding window $W$, a support threshold $m_{c}$, a distance threshold $d$, and a timestamp threshold $k_{c}$, a crowd $Cr$ is a sequence of clusters at non-consecutive timestamps, i.e., $Cr$ =$\langle c_{t_a} ,c_{t_b},\dots, c_{t_k} \rangle$ ($W.s\leq t_a < t_{b}<\dots< t_k\leq W.e$), which satisfies the following three requirements:\\
(1) The number of clusters in $Cr$, i.e., the number of timestamps, is not less than $k_c$.\\
(2) There should be at least $m_c$ objects in each cluster of $Cr$.\\
(3) The distance between any adjacent pair of clusters in $Cr$ is not greater than $d*\Delta t$, where $\Delta t$ is time difference between the pair of clusters.
\end{Def}

Intuitively, a crowd is bounded in a sliding window, so $k_c$ should be less than $w$. If $k_c$ is equal to $w$, our ``crowd'' degenerates to the $crowd$ in~\cite{zheng2013discovery}. We also use the Hausdorff distance~\cite{huttenlocher1993comparing} to measure the distance between two clusters.
Given two clusters $c_1$ and $c_2$, the Hausdorff distance $d_H(c_1,c_2)$ between them is defined as: 
\begin{equation}
d_{H}(c_1,c_2) = \max \{\max\limits_{p\in c_1}\min\limits_{q\in c_2}d(p,q),\max\limits_{q\in c_2}\min\limits_{p\in c_1}d(p,q)\}
\end{equation}

Moreover, a crowd $Cr$ is said to be closed iff there is no superset of $Cr$ which is a crowd in the current window. Our goal is to find the closed crowds to avoid exploring redundant crowds. Essentially the concept of crowd in sliding window can capture dense group of object clusters in most recent time.
Unlike $crowd$ in~\cite{zheng2013discovery},
we do not require that clusters in a crowd are consecutive, i.e., our ``crowd'' has a relaxed time restriction.
Furthermore, the adjacent clusters in a crowd should satisfy the condition that their distance should be less than $d*\Delta t$, namely, we enlarge the distance threshold by being proportional to time difference to connect the clusters at the non-consecutive timestamps reasonably. Hausdorff distance obeys metric properties~\cite{huttenlocher1993comparing}, i.e., Hausdorff distance has the properties of identity, symmetry, and triangle inequality.
This guarantees that the subset of $Cr$ is also a crowd if the length is not less than $k_c$, meaning that crowd also satisfies the downward closure property.
Before defining the evolve group, we define the notions of participator and aggregation in sliding window environment first.

\begin{Def}[Participator]
\label{Participator}
Given a crowd $Cr$ in the current window $W$ and a cluster support threshold $k_p$, an object $o$ is called a participator of $Cr$ iff it appears in at least $k_p$ clusters of $Cr$.
\end{Def}

\begin{Def}[Aggregation]
\label{Gather}
Given a crowd $Cr$ in the current window $W$ and participator count threshold $m_p$, $Cr$ is called a aggregation iff each cluster of $Cr$ includes at least $m_p$ participators.
\end{Def}

By Definition~\ref{Participator}, a participator need not stay in each cluster of the crowd. As long as an object occurs in the crowd at enough timestamps, it is regarded as a participator. An aggregation also need not require a participator to occur at each timestamp in the window. 

If a crowd $Cr$ is an aggregation in the current window $W_c$ and there is no super-crowd $Cr' \supset Cr$ that is an aggregation, then $Cr$ is a closed aggregation in $W_c$.
However, even if some clusters of a crowd $Cr$ do not include enough participators, its sub-crowd may still be an aggregation if the sub-crowd satisfies the constraint of $m_p$ participators in each cluster.
Therefore, in each window we may still need to detect the closed aggregations by exploring sub-crowds space although we aim to find the closed crowds to reduce redundant crowd searching.

\begin{Def}[Group]
\label{Group}
Given an aggregation $Cr$ in the current window $W$, the set of all participators in $Cr$ is called a group that corresponds to $Cr$ in $W$.
\end{Def}

There is a one-to-one correspondence between a group and the aggregation in which it appears.
We can also refer the concept of group to the \textit{core members} of the aggregation in current window.
Next we introduce the concept of evolving group that is exactly the problem which this paper studies.


\begin{Def}[Evolving Group]
\label{EvolvingGroup}
Given a group $Gr_1$ in $W_i$, a group $Gr_2$ in $W_{i+1}$, a fraction support threshold $m_g$, $Gr_{2}$ is evolved from $Gr_1$ in streaming environments iff $|Gr_1\cap Gr_2|\geq m_g*Min(|Gr_1|,|Gr_2|)$, denoted as $\langle Gr_1, Gr_2 \rangle$.
Given a lifetime support threshold $k_g$, an evolving group is a chain of groups during at least $k_g$ consecutive windows, i.e.,  $\langle Gr_{t_a}, Gr_{t_{a+1}}, \dots, Gr_{t_b} \rangle$, where $Gr_{t_{i+1}}$ is evolved from $Gr_{t_{i}}$ $(i=a,a+1,\dots, b-1)$ and $|b-a+1|\geq k_g$.
\end{Def}

Essentially, evolving group can capture the change of \textit{core members} of aggregations over time, i.e., the development of group patterns and causal relationships of groups between adjacent windows, which provides great opportunities for analyzing the impact of the different types of group events on current situation in stream context.


\begin{figure}[h]
\begin{center}
\includegraphics[width=0.85\linewidth]{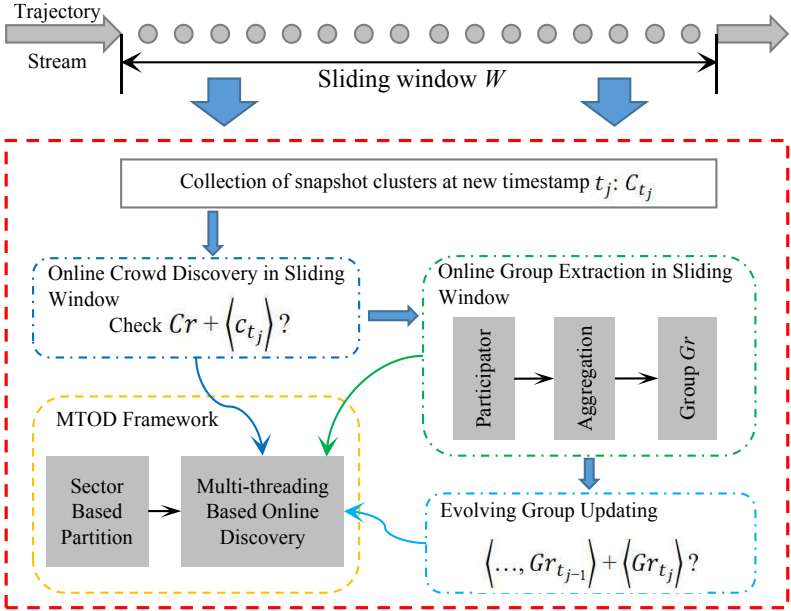}
\caption{Architectures of DEG and MTOD frameworks.}
\label{fig.framework}
\end{center}
\end{figure}

\section{DEG: Discovering Evolving Groups}
\label{sec.alg}
In this section, we now present our DEG method for \underline{d}iscovering all \underline{e}volving \underline{g}roups over trajectories in streaming window environment. 
Figure~\ref{fig.framework} illustrates the overall architectures of DEG and MTOD frameworks.
Basically, our DEG method includes three phases: online crowd discovery, online group detection, and evolving group updating.
In the first phase, we first perform density-based clustering (DBSCAN) on the trajectory points of objects to find all the snapshot clusters at new timestamp. Thus we get the set of snapshot clusters $C_{t_j}$ = $\{c_{t_j}^1 ,c_{t_j}^2,\dots,c_{t_j}^n \}$ at timestamp $t_j$. Then we try to find all closed crowds based on the clusters in the current window.
The second phase aims to validate each closed crowd to see if it is or contains a closed aggregation, and then extract the corresponding group of core members. In third phase, the evolving groups are checked and updated as the window slides. 
MTOD framework is presented in Section~\ref{sec.mtod}.

\subsection{Online Crowd Discovery in Sliding Window}
We first introduce the incremental algorithm, called \textbf{Incre}, which discovers the closed crowds in sliding window environment. It first utilizes the obtained clusters (by DBSCAN) to construct the sequences of clusters at new timestamp. Then the sequences of clusters are used as candidates to validate if they are closed. By leveraging overlap of adjacent windows in sliding window, our incremental algorithm successfully avoids the redundant crowd searching at previous timestamp.
Using an example-driven approach, we now describe how Incre algorithm detects the closed crowds in sliding window.

\begin{table}[h]
\begin{center}
\caption{Example 1. snapshot clusters in sliding window}
\label{exam1}
\begin{tabular}{p{0.8cm}|p{0.8cm}|p{0.8cm}|p{0.8cm}|p{0.8cm}}
\hline
$t_{1}$ & $t_{2}$ & $t_{3}$ & $t_{4}$ & $t_{5}$  \\
\hline\hline
        & $c_2^1$ &         &         & $c_5^2$ \\
\hline
        &         & $c_3^1$ &         &         \\
\hline
$c_1^1$ &         &         & $c_4^1$ & $c_5^1$ \\
\hline
$c_1^2$ &         &         &         &         \\
\hline
\end{tabular}
\end{center}
\end{table}

\begin{example}
\label{exap1}
We use the example in Table~\ref{exam1} to illustrate the discovery of closed crowds in sliding window. To keep its simplicity, we assume the two clusters in the same row or adjacent rows are close to each other, i.e., their Hausdorff distance is not greater than $d*\Delta t$.
\end{example}

\begin{table}[htbp]
\newcommand{\tabincell}[2]{\begin{tabular}{@{}#1@{}}#2\end{tabular}}
\begin{center}
\caption{Illustration of closed crowd discovery}
\label{crowds}
\begin{tabular}{c|c|c|c|c}
\hline
Timestamp  & $CanSet$ & $endclu$ & $status$  & $CloCr$\\
\hline
  $t_1$     & $\langle c_1^1\rangle$;$\langle c_1^2\rangle$ &  $c_1^1$;$c_1^2$  &    &    \\
\hline
  $t_2$     & $\langle c_2^1\rangle$;$\langle c_1^1\rangle$;$\langle c_1^2\rangle$ & $c_2^1$;$c_1^1$; $c_1^2$ & \tabincell{c}{$c_1^1.st$=$unmth$ \\ $c_1^2.st$=$unmth$ } & \\
\hline
  $t_3$     &\tabincell{c}{ $\langle c_2^1, c_3^1\rangle$; \\ $\langle c_1^1, c_3^1\rangle$; \\ $\langle c_2^1\rangle$; $\langle c_1^1\rangle$; $\langle c_1^2\rangle$} &$c_3^1$;$c_2^1$; $c_1^1$;$c_1^2$ &\tabincell{c}{$c_2^1.st$=$match$ \\ $c_1^1.st$=$match$ \\$c_1^2.st$=$unmth$} &  \\
\hline
  $t_4$     &\tabincell{c}{$\langle c_2^1, c_3^1, c_4^1 \rangle$; \\ $\langle c_1^1, c_3^1, c_4^1\rangle$; \\$\langle c_1^2, c_4^1\rangle$; \\ $\langle c_2^1, c_3^1\rangle$; \\ $\langle c_1^1, c_3^1\rangle$; \\ $\langle c_2^1\rangle$; $\langle c_1^1\rangle$; $\langle c_1^2\rangle$} &$c_4^1$;$c_3^1$; $c_2^1$;$c_1^1$; $c_1^2$ &\tabincell{c}{$c_3^1.st$=$match$ \\ $c_2^1.st$=$match$ \\ $c_1^1.st$=$match$ \\$c_1^2.st$=$match$} & \tabincell{c}{$\langle c_2^1, c_3^1, c_4^1\rangle$;\\ $\langle c_1^1, c_3^1, c_4^1\rangle$}\\
\hline
\sout{$t_1$} & \tabincell{c}{$\langle c_2^1, c_3^1, c_4^1 \rangle$; \\ \sout{$\langle c_3^1, c_4^1\rangle$;} \\ \sout{$\langle c_4^1\rangle$;}\\ $\langle c_2^1, c_3^1\rangle$; \sout{$\langle c_3^1\rangle$;} \\ $\langle c_2^1\rangle$ } & $c_4^1$;$c_3^1$; $c_2^1$ &   & \\
\hline
\multirow{2}{*}{$t_5$} &  \multirow{2}{*}{\tabincell{c}{$\langle c_2^1,c_3^1, c_4^1 , c_5^1 \rangle$; \\$\langle c_2^1, c_3^1, c_5^2\rangle$; \\ $\langle c_2^1,c_3^1, c_4^1 \rangle$; \\ \sout{$\langle c_2^1, c_3^1\rangle$;}\\ \sout{$\langle c_2^1\rangle$} }}& \multirow{2}{*}{$c_5^1$;$c_5^2$; $c_4^1$;$c_3^1$; $c_2^1$ } & \tabincell{c}{$c_4^1.st$=$match$ \\ $c_3^1.st$=$match$ \\ $c_2^1.st$=$match$} & \tabincell{c}{$\langle c_2^1, c_3^1, c_4^1, c_5^1 \rangle$;\\ $\langle c_2^1, c_3^1, c_5^2 \rangle$}
  \\ \cline{4-4}
 & & & \tabincell{c}{$c_4^1.st$=$unmth$ \\ $c_3^1.st$=$match$ \\ $c_2^1.st$=$match$} &\\
\hline
\end{tabular}
\end{center}
\end{table}

We first introduce the data structures used in our algorithm to discover closed crowds in streaming windows.
We use $CanSet$ to store all closed crowd candidates in the current window $W_c$. To support closure check, each candidate $Cr$ maintains a $Cr.endclu$ to indicate the ending cluster of $Cr$ in $W_c$. For each ending cluster, we need to validate whether there exists a new cluster which can be appended to $Cr$ in the new window. Therefore, a status flag of $endclu$, denoted as $endclu.st$, is maintained for each ending cluster, which is set to $uncheck$ initially at the beginning of each new window. If there exists a new cluster $c_{new}$ that can be appended to the ending cluster $endclu$, then $endclu.st$=$match$, otherwise $endclu.st$ is set to unmatch (abbr. $unmth$). We use $CloCr$ to denote the closed crowds at each window.
In particular, it is easy to observe that a closed crowd in the current window can be directly checked in the next window according to the following lemma.

\begin{Lem}
\label{lem1}
Given a closed crowd $Cr = \langle c_{t_{i+a}}, \dots, c_{t_{j-b}}\rangle$ $(a\geq 0, b\geq 0)$ in window $W_{t_i}^{t_j}$, if $\exists c_{t_{j+1}}$ in window $W_{t_{i}}^{t_{j+1}}$ such that $Cr_{new} = \langle c_{t_{i+a}}, \dots, c_{t_{j-b}}, c_{t_{j+1}}\rangle$ is a new crowd, then $Cr_{new}$ is a closed crowd in $W_{t_{i}}^{t_{j+1}}$.
\end{Lem}


By Lemma \ref{lem1}, we can discover the closed crowds in the first window by incrementally validating the clusters at each new timestamp.
Table \ref{crowds} shows the intermediate status of crowds in Example~\ref{exap1} at each timestamp.
Suppose $w=4$ and $k_c=3$. There are two clusters $c_1^1$ and $c_1^2$ at time $t_1$, hence two candidates $\langle c_1^1\rangle$ and $\langle c_1^2\rangle$ are stored in $CanSet$, and the ending clusters are inserted into $endclu$. After $t_2$ arrives, we get a cluster $c_2^1$, however, $c_2^1$ is very far away from the clusters at $t_1$. Therefore, no crowd candidate is updated, and the statuses of $c_1^1$ and $c_1^2$ are labeled to $unmth$.
At $t_3$, there is a cluster $c_3^1$ being close to $c_2^1$ and $c_1^1$, thus $\langle c_2^1, c_3^1\rangle$ and $\langle c_1^1, c_3^1\rangle$ are inserted into $CanSet$. $c_2^1.st$ and $c_1^1.st$ are set to $match$, while $c_1^2.st$ is set to $unmth$. However, the candidates of last timestamp are still stored in $CanSet$ for validating potential crowd candidates in next timestamp.
After $t_4$ arrives, we find that $c_4^1$ is close with $c_3^1$, thus $\langle c_2^1, c_3^1, c_4^1\rangle$ and $\langle c_1^1, c_3^1, c_4^1\rangle$ are generated as closed crowd candidates in window $W_{t_1}^{t_4}$. Hence $c_3^1.st$=$match$.
Moreover, we also observe that $\langle c_1^1, c_4^1\rangle$ is already included in $\langle c_1^1, c_3^1, c_4^1\rangle$ although $c_1^1$ is also close to $c_4^1$ w.r.t. $d*\Delta t$. $\langle c_2^1, c_4^1\rangle$ is also included in $\langle c_2^1, c_3^1, c_4^1\rangle$ even assuming that $c_2^1$ was close to $c_4^1$. Namely, any other candidates in $CanSet$ ending with $c_2^1$ or $c_1^1$ must not form a closed crowd in the current window.
Therefore, we set $c_2^1.st$ and $c_1^1.st$ to $match$ directly instead of validating the candidates with $c_4^1$ again. We can further deduce the following corollary.

\begin{Cor}
\label{cor1}
Given a crowd candidate $Cr = \langle c_{t_{i+a}}, \dots, c_{t_{j-b}}\rangle$ $(a\geq 0, b\geq 1)$ in window $W_{t_i}^{t_j}$, if $c_{t_{j-b}}$ is already included in a closed crowd candidate in $W_{t_i}^{t_{j+1}}$, then $Cr +\langle c_{t_{j+1}} \rangle$ must not form a closed crowd in $W_{t_i}^{t_{j+1}}$.
\end{Cor}

Then we can quickly detect all closed crowds in $W_{t_1}^{t_4}$. Namely, the candidates that end with clusters at $t_4$ or clusters whose status are $unmth$, and have number of clusters not less than $k_c(3)$. So $\langle c_2^1, c_3^1, c_4^1\rangle$ and $\langle c_1^1, c_3^1, c_4^1\rangle$ are reported as closed crowds in final.

After $t_5$ arrives, the window slides from $W_{t_1}^{t_4}$ to $W_{t_2}^{t_5}$. Since time $t_1$ has expired, all clusters at $t_1$ are removed from $CanSet$ and $endclu$, as shown in Table \ref{crowds}.
But we can see that some sub-crowds of existing candidates incur due to the remove of $t_1$. \textit{Which sub-crowds should be removed? And which sub-crowds should be reserved?} We observe that the sub-crowd that shares the ending clusters with its super-crowd should be removed from $CanSet$ because it must be included in a close crowd candidate in next window. On the other hand, the sub-crowd that does not include the ending cluster of its super-crowd should be reserved since it may form a close crowd candidate with new cluster in next window.
Therefore, $\langle c_2^1, c_3^1\rangle$ is reserved, while other sub-crowd candidates are removed from $CanSet$.
At $t_5$, there are two clusters $c_5^1$ and $c_5^2$. For validating these two clusters, we need to store two copes of $CanSet$, $endclus$ and their status. For $c_5^1$, we generate a new crowd candidate $\langle c_2^1, c_3^1, c_4^1, c_5^1\rangle$ since $c_5^1$ is near $c_4^1$, and the statuses of $c_4^1$, $c_3^1$ and $c_2^1$ are labeled to $match$.
For $c_5^2$, since $c_5^2$ is far from $c_4^1$ but is close to $c_3^1$, we generate a new crowd candidate $\langle c_2^1, c_3^1, c_5^2 \rangle$ to insert into $CanSet$, and set $c_4^1.st$ to $unmth$, $c_3^1.st$ and $c_2^1.st$ to $match$.

Again, we can output the closed crowds from $CanSet$ in $W_{t_2}^{t_5}$, $\langle c_2^1, c_3^1, c_4^1, c_5^1 \rangle$ and $\langle c_2^1, c_3^1, c_5^2 \rangle$, namely, candidates ending with clusters at $t_5$ or clusters whose $st$ is $unmth$ at all copies.

\subsubsection{Optimization strategies.}

We next present two optimization strategies to minimize computation cost. It is easy to observe that the number of $CanSet$ has a great impact on cost of our Incre algorithm. We first propose the Lemma~\ref{lem2} to reduce the unnecessary maintained candidates.

\begin{Lem}
\label{lem2}
Given a crowd candidate $Cr = \langle c_{t_{i+a}}, \dots, c_{t_{j-b}}\rangle$ $(a\geq 0, b\geq 0)$ in window $W_{t_i}^{t_j}$, and support threshold $k_c$, if $b>j-i+1-k_c$, then $Cr$ is not a crowd or a part of a crowd from $W_{t_{i}}^{t_{j}}$ till to expiration of $Cr$.
\end{Lem}

\begin{proof}
By Definition~\ref{Def-crowd}, the number of clusters of a crowd in a window is not less than $k_c$, thus the number of timestamps at which absent clusters in the crowd is not greater than $w-k_c$. Here, $w=j-i+1$. So if $b>j-i+1-k_c$, there are no clusters in $Cr$ for at least $w-k_c$ timestamps, namely $Cr$ is not a crowd in $W_{t_{i}}^{t_{j}}$.

Moreover, $Cr$ will not be updated to a crowd in next window since the gap between $c_{t_{j-b}}$ and its next connected cluster must be greater than $w-k_c$. Therefore, $Cr$ will never be a part of a crowd in future windows even if all clusters in $Cr$ expire with window sliding.
\end{proof}

Applying Lemma~\ref{lem2}, the optimized $CanSet$ in window $W_{t_2}^{t_5}$ is shown as last row in Table~\ref{crowds}.

Furthermore, we also observe that the order of candidates in $CanSet$ and their corresponding ending clusters can improve the efficiency of our algorithm by pruning the unnecessary verification for new clusters in the new window.

\begin{Lem}
\label{lem3}
Given a crowd candidate $Cr = \langle c_{t_{i+a}}, \dots, c_{t_{j-b}}\rangle$ $(a\geq 0, b\geq 0)$ in window $W_{t_i}^{t_j}$, and a new cluster $c_{t_{j+1}}$, if $d_H(c_{t_{j-b}},c_{t_{j+1}})\leq d*(b+1)$, then there is no need to validate the candidates that end with any $c$ in $Cr$-$\langle c_{t_{j-b}} \rangle$ with $c_{t_{j+1}}$ again.
\end{Lem}

Lemma \ref{lem3} can be easily proved by Corollary \ref{cor1}.
Based on the observation, we sort $CanSet$ in the last time first order by ending clusters, and also store the corresponding ending clusters in last time first order in $endclu$. For validating a new cluster $c_{new}$, we traverse the candidates with the latest time ending cluster, if $c_{new}$ is appended into a candidate $Cr$, all clusters in $Cr$ are set to $match$ if they are in $endclu$ copy, i.e., we will skip the candidates that end with the clusters in next traversal on $CanSet$ for $c_{new}$.

\subsubsection{Cluster pruning strategy.}
Indexing clusters with R-tree or grid can improve the efficiency of discovery algorithm~\cite{zheng2014online}. However the indexing methods suffer from two major drawbacks. First, indexing does not work very well in dynamic streaming environments in which the index has to be continuously rebuilt when streaming data evolves. Second, it is also very expensive for mapping the clusters that include lots of objects into index. Here, we propose a cluster pruning strategy to reduce the number of Hausdorff distance computation.

Specifically, we use the mean center of cluster $m_i$ and maximum radius $r_i$ to represent the cluster $c_i$. $r_i$ is the distance from the center $m_i$ to the farthest point in $c_i$. Intuitively, we get following two pruning rules.

\begin{rul}[Long-distance Pruning]
Consider two clusters $c_i$ and $c_j$, their mean centers $m_i$ and $m_j$, radius $r_i$ and $r_j$, and time difference $\Delta t$, if the distance between $m_i$ and $m_j$ is greater than $r_i+r_j+d*\Delta t$, then the Hausdorff distance between $c_i$ and $c_j$ must be greater than $d*\Delta t$.
\end{rul}

\begin{rul}[Short-distance Pruning]
Consider two clusters $c_i$ and $c_j$, their mean centers $m_i$ and $m_j$, radius $r_i$ and $r_j$, and time difference $\Delta t$, if the distance between $m_i$ and $m_j$ is not greater than $d*\Delta t-\max\{r_1,r_2\}$, then the Hausdorff distance between $c_i$ and $c_j$ must not be greater than $d*\Delta t$.
\end{rul}

\begin{figure}[htbp]
\centering
\subfigure[Long-distance pruning]{
\label{fig.prune1}
\includegraphics[width=0.46\textwidth]{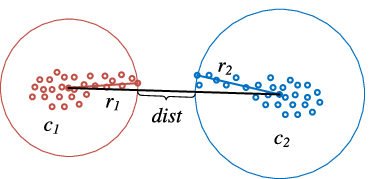}}
\hspace{5mm}
\subfigure[Short-distance pruning]{
\label{fig.prune2}
\includegraphics[width=0.36\textwidth]{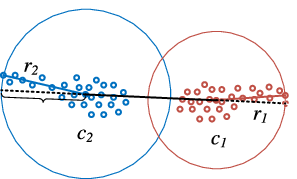}}
\caption{Examples of pruning cluster}
\label{fig.prune}
\end{figure}

As shown in Fig.~\ref{fig.prune1}, the smallest minimum distance from points in $c_1$ to $c_2$ must be larger than $dist$, similarly, the smallest minimum distance from points in $c_2$ to $c_1$ also must be larger than $dist$. Thus $d_H(c_1,c_2)>dist$. Therefore, if $|m_1-m_2|>r_1+r_2+d*\Delta t$, namely, $dist>d*\Delta t$, then $d_H(c_1,c_2)>dist>d*\Delta t$.

In Fig.~\ref{fig.prune2}, the maximum minimum distance from points in $c_2$ to $c_1$ must be less than $r_2+|m_1-m_2|$, similarly, the maximum minimum distance from points in $c_1$ to $c_2$ must be less than $r_1+|m_1-m_2|$. Hence, $d_H(c_1, c_2)<\max\{r_1,r_2\}+|m_1-m_2|$. If $|m_1-m_2|< d*\Delta t-\max\{r_1,r_2\}$, namely, $\max\{r_1,r_2\}+|m_1-m_2|<d*\Delta t$, then $d_H(c_1, c_2)<d*\Delta t$.

By the pruning rules, we first calculate the distance between the mean centers to determine whether we need to continue to compute the Hausdorff distance for validating the final results.
Obviously, only the clusters that are apart between $d*\Delta t-\max\{r_i,r_j\}$ and $r_i+r_j+d*\Delta t$ need to be calculated the Hausdorff distances.

\subsubsection{Incre algorithm.}

Algorithm~\ref{alg-crowd} shows the pseudo-code of our Incre algorithm. First, Incre removes the information of $t_{i-1}$ from $CanSet$ of last window. In particular, we also use Lemma \ref{lem2} to remove the unnecessary maintained candidates (lines 4-5). $Cr.endclu.t$ denotes the timestamp of ending cluster of candidate $Cr$.
For a new cluster $c_{t_j}$, we maintain a copy of $CanSet$, denoted as $c_{t_j}.CanSet$. Each candidate whose ending cluster is unchecked would be validated with $c_{t_j}$, as shown in lines 9-16. If $c_{t_j}$ is close to $Cr.endclu$, a new candidate is generated and inserted into $CanSet$, and all ending clusters in $Cr$ are set to $match$ (lines 11-14). $Cr.allendclu$ stands for all ending clusters in $Cr$.
Finally, all length enough ($\geq k_c$) candidates that end with clusters at $t_j$ or with $numth$ status in all copies are reported as closed crowds (lines 17-19). $Cr.endclu.allst=unmth$ denotes the statuses of $Cr.endclu$ are $unmath$ in $CanSet$ copies of all new clusters.

\begin{algorithm}[htbp]
\caption{Incre algorithm for Crowd Discovery}
\label{alg-crowd}
\renewcommand{\algorithmicrequire}{\textbf{Input:}}
\renewcommand{\algorithmicensure} {\textbf{Output:}}
\begin{algorithmic}[1]
\Require{the current window $W_{t_i}^{t_j}$, $C_{t_j}$, $k_{c}$, $m_{c}$, $d$}
\Ensure {Closed crowds $CloCr$}
\State $CloCr\leftarrow \emptyset$; $CanSet \leftarrow W_{t_{i-1}}^{t_{j-1}}.CanSet$;
\For {each $Cr \in CanSet$}
    \State Remove $c_{t_{i-1}}$ from $Cr$; //delete clusters at time $t_{i-1}$
    \If {($|t_{j-1}-Cr.endclu.t|>j-i+1-k_c$)}
        \State Remove $Cr$ from $CanSet$;
    \EndIf
\EndFor
\State $CanSetCopy \leftarrow CanSet$;
\For {each $c_{t_j}\in C_{t_j}$}
\State $c_{t_j}.CanSet \leftarrow CanSetCopy$; // a copy of $CanSet$
    \For {each $Cr \in c_{t_j}.CanSet$}
        \If {($Cr.endclu.st == uncheck$)}
            \If {($d_H(c_{t_j},Cr.endclu)\leq d*\Delta t$)}
                \State Insert $Cr+\langle c_{t_j}\rangle$ into $CanSet$;
                \State $Cr.endclu.st\leftarrow match$;
                \State $Cr.allendclu.st\leftarrow match$;
            \Else
                \State $Cr.endclu.st\leftarrow unmth$;
            \EndIf
        \EndIf
    \EndFor
\EndFor
\For {each $Cr \in CanSet$}
    \If {(($Cr.endclu.t==t_j|Cr.endclu.allst==unmth$) $\& Cr.len\geq k_c$)}
        \State $CloCr \leftarrow CloCr\cup Cr$;
    \EndIf
\EndFor
\end{algorithmic}
\end{algorithm}

\subsection{Online Group Extraction in Sliding Window}

Next, we discuss the algorithm to detect closed aggregation and corresponding group in each obtained closed crowd in sliding window. We also first use an example to elaborate our process method. Then we present the details of our algorithm on basis of the proposed principles.

\begin{table}[h]
\caption{Example 2. a closed crowd in $W_{t_1}^{t_8}$ and $W_{t_2}^{t_9}$}
\label{tab.gather}
\centering
\begin{tabular}{c|c|c|c|c|c|c|c|c}
\hline
$c_{t_1}$&$c_{t_2}$&$c_{t_3}$&$c_{t_4}$&$c_{t_5}$&$t_6$   &$c_{t_7}$&$c_{t_8}$&$c_{t_9}$  \\
\hline
$o_{1}$ & $o_{1}$ & $o_{1}$ &         & $o_{1}$ &         & $o_{1}$ & $o_{1}$ &       \\
\hline
$o_{2}$ & $o_{2}$ &         & $o_{2}$ & $o_{2}$ &         & $o_{2}$ & $o_{2}$ & $o_{2}$\\
\hline
$o_{3}$ & $o_{3}$ &         & $o_{3}$ & $o_{3}$ &         & $o_{3}$ & $o_{3}$ & $o_{3}$\\
\hline
        &         & $o_{4}$ &         &         &         & $o_{4}$ &         &        \\
\hline
        & $o_{5}$ &         & $o_{5}$ & $o_{5}$ &         & $o_{5}$ & $o_{5}$ & $o_{5}$\\
\hline
o$_{6}$ & $o_{6}$ & $o_{6}$ & $o_{6}$ &         &         &         & o$_{6}$ &        \\
\hline
\end{tabular}
\end{table}

\begin{example}
\label{examp2}
Consider a closed crowd shown in Table~\ref{tab.gather}, and let $w=8$, $k_c=6$, $k_p=5$ and $m_p=3$. There is a closed crowd $Cr_1= \langle c_{t_1}, c_{t_2}, c_{t_3}, c_{t_4}, c_{t_5}, c_{t_7}, c_{t_8}\rangle$ with 6 objects in window $W_{t_1}^{t_8}$ that also evolves to a closed crowd $Cr_2= \langle c_{t_2}, c_{t_3}, c_{t_4}, c_{t_5}, c_{t_7}, c_{t_8}, c_{t_9}\rangle$ in $W_{t_2}^{t_9}$.
\end{example}

We first verify if the crowd is an aggregation in window $W_{t_1}^{t_8}$. 
It is easy to see that the objects $\{o_1, o_2, o_3, o_5, o_6\}$ are participators w.r.t. threshold $k_p(5)$. But not every cluster in $Cr_1$ contains $m_p(3)$ participators, e.g., $c_{t_3}$ is considered as an invalid cluster since it only contains two participators. However, $\langle c_{t_1}, c_{t_2}, c_{t_4}, c_{t_5}, c_{t_7}, c_{t_8}\rangle$ still is a crowd in this window by removing the invalid cluster, meaning that it still may be an aggregation. Therefore, we need to further check the sub-crowd whether it is a closed aggregation. 
Again we get participator set $\{o_1, o_2, o_3, o_5\}$ from the sub-crowd. Obviously, the set of participators is a subset of that of $Cr_1$. Now we introduce the following lemma by the observation.

\begin{Lem}
\label{lem4}
Given a crowd $Cr$ and its corresponding participator set $O$ in window $W_{t_i}^{t_j}$, for any crowd $Cr' \subset Cr$, the corresponding participator set $O'$ of $Cr'$ is a subset of $O$.
\end{Lem}

Lemma~\ref{lem4} is intuitive. 
For any participator $o \in O'$ of $Cr'$, there must exist at least $k_p$ clusters in $Cr'$ that contain object $o$. Since $Cr' \subset Cr$, the $k_p$ clusters also must be included in $Cr$. Therefore, $o$ must be a participator of $Cr$, namely, $o \in O$.

Specifically, we use $o.cnt$ to denote the number of clusters in which $o$ appears in a crowd. For example, we get $\{o_1.cnt=6, o_2.cnt=6, o_3.cnt=6, \overline{o_4.cnt=2,}\ o_5.cnt=5, o_6.cnt=5\}$ in $Cr_1$, where the object with an overline denotes a non-participator because its $cnt$ is less than $k_p(5)$. When re-checking the participators in $Cr_1- \langle c_{t_3} \rangle$, we only need to verify if these participators appear in $c_{t_3}$. If yes, $o.cnt$ is reduced by 1, otherwise, $o.cnt$ stays the same. Therefore, we can fast find out the new participator set $\{o_1.cnt=5, o_2.cnt=6, o_3.cnt=6, o_5.cnt=5, \overline{o_6.cnt=4}\}$. Then we detect a closed aggregation $\langle c_{t_1}, c_{t_2}, c_{t_4}, c_{t_5}, c_{t_7}, c_{t_8} \rangle$ along with its group $\{o_1, o_2, o_3, o_5\}$, since each cluster contains at least $m_p (3)$ participators. 

After $t_9$ arrives, $W_{t_1}^{t_8}$ slides to $W_{t_2}^{t_9}$. A new closed crowd $Cr_2$ in $W_{t_2}^{t_9}$ is generated based on $Cr_1$ in last step.
First, we need to delete the information of $t_1$ from $Cr_1$. Similarly, we can directly use the above method to fast update participator set by removing invalid cluster $c_{t_1}$ from $Cr_1$. So the participator set is updated to $\{o_1.cnt=5, o_2.cnt=5, o_3.cnt=5, \overline{o_4.cnt=2,}\ o_5.cnt=5, \overline{o_6.cnt=4}\}$.
Next, we consider how to update new cluster $c_{t_9}$ into the obtained result using minimal computation. 
We here add the new cluster into $Cr_1-\langle c_{t_1}\rangle$, namely, we verify $Cr_2$ in basis of the obtained participator set.
Therefore, we only need to check whether the objects appear in new cluster $c_{t_9}$. Thus, the participator set is updated to $\{o_1.cnt=5,\ o_2.cnt=6,\ o_3.cnt=6,\ \overline{o_4.cnt=2,}\ o_5.cnt=6,\ \overline{o_6.cnt=4}\}$.

Similarly, we find that $c_{t_3}$ also does not satisfy the condition of $m_p$ participators, i.e., an invalid cluster, hence we next continue to verify the candidate using above removing method.
Finally, $\langle c_{t_2}, c_{t_4}, c_{t_5}, c_{t_7}, c_{t_8}, c_{t_9} \rangle$ is reported as a closed aggregation with corresponding group $\{o_2, o_3, o_5\}$ in $W_{t_2}^{t_9}$.

Moreover, if $c_{t_9}$ was an invalid cluster, we can see that the detection falls back to obtained result in last window, i.e., $W_{t_2}^{t_8}$. Namely, we only need to directly verify the obtained result in $W_{t_1}^{t_8}$ by removing $t_1$.

By the illustration of Example~\ref{examp2}, we propose our \textbf{Verification with Removing and Adding} (VRA) algorithm to detect closed aggregations with corresponding groups in sliding window efficiently.

As shown in Algorithm~\ref{VRA}, $Ag$ denotes a closed aggregation, and $Gr$ is its corresponding group. $Par$ is the participator set with their $cnt$s. If $Cr$ is an emerging closed crowd, we get its participator set of $Cr$ from scratch (denoted as $Participator(Cr)$). Otherwise, we can fast find out the participator set of $Cr$ from $Par$ of last window by only verifying the objects in $c_{t_{i-1}}$ and $c_{t_j}$ (lines 4-6).
Lines 8-19 show the detection process of closed aggregation applying a downward method.
VRA first checks if each cluster in the crowd copy $Cr'$ contains enough participators shown as lines 12-14. If not, VRA then updates the participators by removing the invalid clusters (lines 9-11), and re-checks each clusters of the remainder again (lines 12-14) till the remainder is not a crowd or there is no more invalid cluster.
If there is no more invalid cluster, $Cr'$ and its current participators are reported as a closed aggregation and corresponding group.
Actually, we only focus on the real participators whose $cnt \geq k_p$ instead of all objects in $Cr$ in this detection process because the participators of a crowd must be from that of its super-crowd by Lemma 4. $Par'(k_p)$ denotes the set of real participators that satisfy $k_p$ threshold, and $Par'(k_p,c)$ denotes the set of real participators in cluster $c$.

However, if $c_{t_j}$ is an invalid cluster in the process, the detection falls back to obtained result in last window. As shown in lines 17-19, we further verify the obtained aggregation of last window by removing the expired timestamp $t_{i-1}$.

\begin{algorithm}[h]
\caption{Verify with Removing and Adding (VRA)}
\label{VRA}
\renewcommand{\algorithmicrequire}{\textbf{Input:}}
\renewcommand{\algorithmicensure} {\textbf{Output:}}
\begin{algorithmic}[1]
\Require{closed crowd $Cr$ in window $W_{t_i}^{t_j}$, closed crowd $Cr_1$ and $Par$ in $W_{t_{i-1}}^{t_{j-1}}$, $k_c$, $k_{p}$, $m_p$}
\Ensure {aggregation $Ag$, and corresponding group $Gr$}
\State $Ag \leftarrow \emptyset$; $Gr \leftarrow \emptyset$; $C_{un} \leftarrow \emptyset$;
\If {$Cr$ is an emerging closed crowd}
    \State $Par \leftarrow Participator(Cr)$;
\Else
    \State $Par \leftarrow Par-obj(c_{t_{i-1}})$; //delete information of $t_{i-1}$
    \State $Par \leftarrow Par+obj(c_{t_{j}})$; //add information of $t_{j}$
\EndIf
    \State $Cr' \leftarrow Cr$; $Par' \leftarrow Par$;
    \While {($|Cr'-C_{un}|\geq k_c$)}
        \For {each $c \in C_{un}$}
            \State $Par' \leftarrow Par'-obj(c)$;
        \EndFor
        \State $Cr' \leftarrow Cr'-C_{un}$; $C_{un} \leftarrow \emptyset$;
        \For {each $c \in Cr'$}
            \If {($|Par'(k_p,c)|<m_p$)}
                \State $C_{un} \leftarrow C_{un}\cup c$;
            \EndIf
        \EndFor
        \If {($C_{un}==\emptyset$)}
            \State $Ag \leftarrow Cr'$; $Gr \leftarrow Par'(k_p)$; Break;
        \ElsIf {($c_{t_j} \in C_{un}$ \& $|Cr'-C_{un}|\geq k_c$)}
            \State Fall back to $Ag$ and $Gr$ in $W_{t_{i-1}}^{t_{j-1}}$;
            \State Verify $Ag-\langle c_{t_{i-1}}\rangle$ on $Gr$; Break;
        \EndIf
    \EndWhile
\end{algorithmic}
\end{algorithm}

\begin{Lem}
\label{lem5}
Given a closed crowd $Cr$ in the current window $W_{t_i}^{t_j}$, If $Ag$ is reported as an aggregation from $Cr$ by $VRA$ algorithm, then the aggregation $Ag$ is closed in $W_{t_i}^{t_j}$.
\end{Lem}

\begin{proof}
Lemma \ref{lem5} can be proved easily. Suppose $Ag$ is reported as an aggregation by $VRA$ algorithm. According to the downward flow of our $VRA$ algorithm, $Ag$ is the biggest sub-crowd of $Cr$ that satisfies the condition that each cluster contains $m_p$ participators. Namely, there is no super-crowd of $Ag$ which is an aggregation. Therefore, $Ag$ is closed in $W_{t_i}^{t_j}$ by closure property of aggregation described in Sec.~\ref{sec.def}.
\end{proof}

\subsection{Evolving Group Updating}

In the third phase, we update the evolving groups by constructing the sequence chain of groups incrementally. For each group obtained in the current window, we first evaluate whether it is evolved from a group of last window, namely, whether it shares most of core members with the groups detected in last window. If yes, we update the group into the evolving groups. Therefore, we can find that each pair of adjacent groups in an evolving group has an evolutionary relationship. If the group does not share most of members with any group in last window, i.e., the group is an emerging aggregation in the current window, a new evolving group that only contains the group is constructed.
Moreover, If the last group in a chain (the group in last window) is not evolved in current window, we say that the evolving group is a closed evolving group.





\subsection{Complexity Analysis}
\label{subsec.complex}
Suppose that there are, on average, $m$ new clusters per timestamp, $n$ ending clusters in $endclu$ and $l$ closed crowd candidates in $CanSet$ per window, and $k$ trajectory points per cluster. The Incre algorithm constructs the closed crowd candidates by computing the Hausdorff distance between new clusters and ending clusters. Thus the complexity of Incre is $O(mnk^2)$ in the worst case. Suppose that the ratio of pruning clusters, on average, $\alpha$ in each window. The complexity becomes $O((1-\alpha)mnk^2)$ after applying the cluster pruning strategy.
For each closed crowd, VRA algorithm detects the closed aggregation by verifying at most $w-k_c$ clusters to update the participator set. All $l$ crowd candidates may be closed in the worst case. Therefore, it worst-case complexity is $O(l(w-k_c)k)$.
For evolving group updating, we evaluate whether the new groups are evolved from the groups of last window. Therefore, the complexity is $O(l^2)$ in the worst case of all $l$ crowd candidates being aggregations.
The overall complexity of our DEG method is $O((1-\alpha)mnk^2)+O(l(w-k_c)k)+O(l^2)$ in the worst case.

\section{Online Discovery Framework}
\label{sec.mtod}
To further drive down computation cost, we now present our \underline{m}ulti-\underline{t}hreading based \underline{o}nline \underline{d}iscovery (MTOD) framework. By proposing sector-based partition for object clusters, the MTOD framework achieves memory-sharing parallel discovery of evolving groups over massive-scale trajectory streams in near real time.

\subsection{Sector-Based Partition}
\label{subsec.secpar}
To further reduce the number of Hausdorff distance computation and simultaneously achieve load balancing, we propose an flexible sector-based partition of snapshot clusters at each timestamp for multi-threading parallel discovery framework.

As shown in Figure~\ref{fig.sector}, we first find the center of data space and use concentric circles to divide the data space into multiple annuluses. The radii of concentric circles are determined according to the distribution of snapshot clusters. Then, in each annulus, we can determine the areas of sectors according to the number of clusters assigned to each thread. 
The partition is updated at each timestamp. For example, in Figure~\ref{fig.sector}, the red solid lines are the dividing lines at the current timestamp, and the green dotted lines indicate the dividing lines at the last timestamp. The snapshot clusters are new clusters at the current timestamp. 

\begin{figure}
\begin{center}
\includegraphics[width=0.50\linewidth]{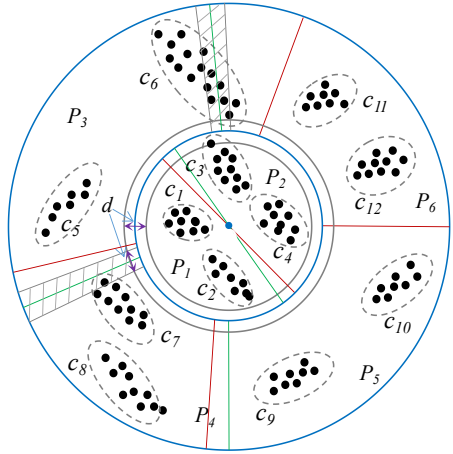}
\caption{An example of sector-based partition. Each gray dashed oval denotes a cluster. The blue circles are concentric circles. The purplish red lines are boundaries of sectors.} 
\label{fig.sector}
\end{center}
\end{figure}

The sector-based partition is more flexible than grid-based partition. we only need to adjust the radii of concentric circles and/or the boundaries of sectors to achieve load balancing at each timestamp. 
As shown in Figure~\ref{fig.sector}, the new partitions can be easily obtained by only updating the sector boundaries at the last timestamp. We use $P_i$ to denote a partition at the current timestamp and $P_i'$ to denote a partition at the last timestamp. 
Notice that the radii of the concentric circles may need to be adjusted when the snapshot cluster distribution changes greatly. 

To this end, we regard the the center of data space as the pole and the horizontal coordinate as the polar axis to construct a polar coordinate system. 
According to the stored mean centers and maximum radii of all new clusters, it is easy to obtain the radii of the concentric circles and the boundary lines of sectors for sector-based partition. 
More specifically, we first sort the polar coordinates of all cluster centers, and then, according to the number of clusters $m$ and the number of threads $T_m$, use the sorted polar radii of cluster centers to determine the radii of concentric circles, and use the sorted polar angles of cluster centers in the same ring to determine the boundary lines of sectors. 

Therefore, the time complexity of re-partition mainly includes the transformation of cluster center coordinates, the polar radius sorting of cluster centers, the polar angle sorting of cluster centers in the same ring, and the division of concentric circles and the determination of sector boundary lines. 
That is, the total time complexity is $O(m + m\log m + \sum m_i\log m_i)$, where $m\log m$ is the time complexity of sorting polar radii of $m$ cluster centers, and $m_i$ is the number of clusters in $i$-th annulus. 
The space complexity is $O(\frac{m}{T_m} + r)$, where $r$ is the number of annuluses, storing the radii of annuluses, and $\frac{m}{T_m}$ is the number of partitions, storing the polar coordinate information of boundary lines of all sectors. 

Additionally, we observe that the new clusters in a partition may not have to validate the ending clusters of last window in other partitions. To guarantee this point, we use the boundary zone to separate the annuluses and sectors. As shown in Figure~\ref{fig.sector}, we use a ring of $2d$ width and a boundary zone of at least $2d$ width to separate two annuluses and two sectors respectively. The new clusters that cover the boundary zone are also checked with the ending clusters of last window in the adjacent partition area. 
In other cases, we never need to validate the closed crowd candidates with the new clusters in other partitions. Namely, we directly prune more clusters before applying the long-distance pruning. 
More specifically, new clusters are classified into two types based on the relative relationship between their locations and sector-based partitions of last timestamp. 
The first type of clusters are those that completely fall within a partition, that is, they have no overlap with any boundary zone. For example, $c_5$ lies completely in $P_3'$ of the last timestamp. Thus $c_5$ is a cluster of the first type. 
The second type of clusters are those remaining clusters, i.e., clusters that have overlap with any boundary zone or cross over any boundary zone. 
For example, since $c_7$ covers the boundary zone of $P_3'$ and $P_4'$, it belongs to the second type. Similarly, $c_6$ crosses over the boundary zone of $P_3'$ and $P_6'$ and also lies on the ring zone of $P_2'$. Thus, $c_6$ is also a cluster of the second type. 
Based on the above analysis, the new clusters of the first type only need to be checked with the ending clusters of last window which lie in the same partition. 
The new clusters of the second type would be checked with the ending clusters of last window which lie in the partitions whose boundary zones overlapping with the new clusters.

The sector-based partition is not only applicable to memory-sharing parallel framework but also applied to shared-nothing distributed architecture, because the partition method can easily group the nearby new clusters with local closed crowd candidates.

\subsection{Theoretic Analysis}

In this section, we theoretically analyze and prove the correctness of our proposed parallel method.


Based on the analysis in Sec.~\ref{subsec.secpar}, we classify the new clusters into two independent categories. Below we will analyze and prove their correctness separately. 
First, we can easily get following Lemma~\ref{lem6}. 

\begin{Lem}
\label{lem6}
Given a new cluster of the first type $c_1$ that completely falls in $P_i'$, for any cluster $c_2$ in $P_j'\ (i \neq j)$, then $c_1$ must have a Hausdorff distance larger than $d$ with $c_2$. 
\end{Lem}


\begin{proof}
Lemma~\ref{lem6} is intuitive and can be proved according to long-distance Pruning in Rule~1. 
Since $c_1$ is a cluster of the first type, it lies completely in $P_i'$ without covering any boundary zone. 
$c_2$ belongs to $P_j'$ ($i \neq j$), thus the distance between $c_1$ and $c_2$ is the smallest when $P_j'$ is an adjacent partition of $P_i'$. 
Suppose that $P_j'$ is an adjacent partition of $P_i'$. 
By the analysis of the long-distance pruning, the $d_H(c_1,c_2)$ must be larger than $d$, because the width of boundary zone between $P_i'$ and $P_j'$ is $2d$. 
That is, even if $c_2$ covers a part of the boundary zone of $P_j'$, it does not cross the boundary line, thus there is still a Hausdorff distance of at least $d$ between $c_1$ and $c_2$. 
\end{proof}

Lemma~\ref{lem6} indicates that a new cluster of the first type is only possible to match the ending clusters of last window which lie in the same partition because all the other clusters definitely have a distance larger than $d$. 

Next, we have following Lemma~\ref{lem7} for the new clusters of the second type. 

\begin{Lem}
\label{lem7} Given a new cluster of the second type $c_1$ that covers the boundary zones among a set of partitions $P_{s}' = \{P_i', P_j',\dots\}$, for any cluster $c_2$ in a partition $P_a'$, if $P_a' \notin P_{s}'$, then $c_1$ must have a Hausdorff distance larger than $d$ with $c_2$. 
\end{Lem}

\begin{proof}
Lemma~\ref{lem7} can be proved based on Lemma~\ref{lem6}. 
Since new cluster $c_1$ only covers the boundary zones of $P_s'$, it will not cover any boundary zone of $P_a'$ because $P_a' \notin P_{s}'$. 
Therefore, $c_1$ has a Hausdorff distance larger than $d$ with $c_2$ in $P_a'$ because of the same reason in Lemma~\ref{lem6}. 
\end{proof} 

Lemma~\ref{lem7} implies that a new cluster of the second type is only possible to match the ending clusters of last window which lie in the partitions whose boundary zone is covered by itself. 
Namely, the new clusters of the second type are impossible to match the ending clusters of last window which lie in the partitions that have no overlap with them. 


Our proposed MTOD framework is mainly based on the sector-based partition to achieve parallel evolving group discovery. For each new cluster of the first type, only the end clusters of last window in the same partition are verified. 
For each new cluster of the second type, only the end clusters of last window in the partitions whose boundary zone is covered by itself are verified. 
Therefore, the correctness of our proposed MTOD framework for the new clusters of the first and second types are both proved through Lemma~\ref{lem6} and Lemma~\ref{lem7}. 

\begin{algorithm}[htbp]
\caption{Multi-Threading based Online Discovery Framework}
\label{alg-mtod}
\renewcommand{\algorithmicrequire}{\textbf{Input:}}
\renewcommand{\algorithmicensure} {\textbf{Output:}}
\begin{algorithmic}[1]
\Require{the current window $W_{t_i}^{t_j}$, $C_{t_j}$, $k_{c}$, $m_{c}$, $d$, $k_p$, $m_p$}
\Ensure {groups $CloGr$ in $W_{t_i}^{t_j}$}
\State $CanSet \leftarrow W_{t_{i-1}}^{t_{j-1}}.CanSet$; 
\For {each $Cr \in CanSet$}
    \State Remove $c_{t_{i-1}}$ from $Cr$; //delete clusters at time $t_{i-1}$
    \If {($|t_{j-1}-Cr.endclu.t| > j-i + 1 - k_c$)}
        \State Remove $Cr$ from $CanSet$;
    \EndIf
\EndFor
\State $CanSetCopy \leftarrow CanSet$;
\State $C^{P_{[1\dots n]}}_{t_j} \leftarrow \textsl{Partition}(C_{t_j})$
\State \textbf{Create} Thread for each partition;
\State $thread\_num \leftarrow n $; 
\State\textbf{begin} \textsl{Multi-Thread process}:
\For {each $c_{t_j}\in C^{P_i}_{t_j}$}
   \State $P_s' \leftarrow LastPartition(c_{t_j})$;
    \State $c_{t_j}.CanSet \leftarrow CanSetCopy(P_s')$;
    \For {each $Cr \in c_{t_j}.CanSet$}
        \If {($Cr.endclu.st == uncheck$)}
            \If {($d_H(c_{t_j},Cr.endclu)\leq d*\Delta t$)}
                \State $Cr.endclu.st\leftarrow match$;
                \State $Cr.allendclu.st\leftarrow match$;
                \If {$Cr.len \ge k_c - 1$}
                	\State $Ag, Gr \leftarrow$ VRA($Cr+\langle c_{t_j}\rangle$)
                \EndIf
                \State \textbf{synchronized}($CanSet$ and $CloGr$)
                \State \qquad Insert $Cr+\langle c_{t_j}\rangle$ into $CanSet$;
                \State \qquad $CanSet.Cr.endclu.st \leftarrow match$;
                \State \qquad $CloGr \leftarrow CloGr \cup Gr$;
            \Else
                \State $Cr.endclu.st\leftarrow unmth$;
            \EndIf
        \EndIf
    \EndFor
\EndFor
\State \textbf{end} \textsl{Multi-Thread process};
\For {each $Cr \in CanSet$}
    \If {($Cr.endclu.st!=macth$ $\&$ $Cr.len \ge k_c$)}
        \State $Ag, Gr \leftarrow$ VRA($Cr$);
        \State $CloGr \leftarrow CloGr \cup Gr$;
    \EndIf
\EndFor
\end{algorithmic}
\end{algorithm}

\subsection{MTOD Framework}

Algorithm~\ref{alg-mtod} shows the MTOD framework. First, we update the closed crowd candidates by removing the information of expired timestamp $t_{i-1}$ and the unnecessary maintained candidates (lines 1-5). Second, MTOD uses sector-based partition method to divide the new clusters into $n$ partitions, and creates one thread for each partition. 
Next, in each thread process, we apply the proposed Incre and VRA to update the closed candidates in $CanSet$ and discover the closed aggregations and the corresponding groups, as shown in lines 10-27. 
Finally, for each unmatched closed crowd candidate after multi-thread process, MTOD uses VRA to check if it is a close aggregation (lines 28-31).
Note that $C_{t_j}^{P_i}$ denotes the set of new clusters in partition $P_i$ at the current timestamp. 
Function $LastPartition(c_{t_j})$ returns the partitions of last timestamp covered by $c_{t_j}$. Specifically, if $c_{t_j}$ is a cluster of the first type, $LastPartition(c_{t_j})$ returns the partition that $c_{t_j}$ falls into, otherwise, $LastPartition(c_{t_j})$ returns all the partitions $c_{t_j}$ covers. 
$CanSetCopy(P_s')$ denotes the local closed crowd candidates in all partitions in $P_s'$.

\section{Experiment}
\label{sec.exp}

\subsection{Datasets}
\label{subsec.datasets}
We use three real world datasets to evaluate the effectiveness and efficiency of our proposed algorithms compared against the state-of-the-art.

\textit{Pedestrian data.} The dataset~\cite{zanlungo2014potential}\cite{zanlungo2015spatial}\cite{brscic2013person} contains the pedestrian position and group annotations in the ATC shopping center in Osaka, Japan. The tracking of pedestrians are done using automatic tracking systems, whereas the groups are labeled manually.
In the dataset, there are two types of files, person tracking files and group files. Person tracking files contain the data for all persons that were tracked in the environment on a given day and period of time. Group files contain the group annotations for the given day. Only pedestrians in groups are listed, pedestrians walking alone are not included. The group files can be used to be compared with our detected results and then get \textit{precision} and \textit{recall}. 
The dataset contains 8 experiment days, and the data for 4 one-hour periods (10:00-11:00, 12:00-13:00, 15:00-16:00, and 19:00-20:00) is provided for each day. The time domain is split into 5 seconds of granularity in our experiments.
More details about this dataset can be found in this  website\footnote{http://www.irc.atr.jp/sets/groups/}.

\textit{Taxi data.} 
The dataset is from T-Drive project~\cite{yuan2013t} collected by Microsoft Research Asia.
T-drive is a smart driving direction services based on GPS trajectories of a large number of taxis, which includes real-world trajectories generated by 30,000 taxis in Beijing in a period of 3 months.
In our experiments, we use a sample of the dataset that contains one week trajectories of 10,357 taxis in a period from February 2 to February 8, 2008.
The total number of points in this dataset is about 15 million and the total distance of the trajectories reaches 9 million kilometers.
We divide a day into four time periods, morning and evening peak time (7:00-10:00 and 16:00-20:00), work time in morning and noon (10:00-13:00) and work time in afternoon (13:00 to 16:00). We interpolate the time domain into the granularity of minute on $Taxi$ dataset.

\textit{Traffic data.}
The dataset is citywide surveillance traffic data collected in Jinan, China. 
This dataset contains $405,370,631$ records of total $11,299,927$ vehicles from $1,704$ surveillance cameras over the period of August 1st, 2016 - August 31st, 2016. We also interpolate the time domain into the granularity of minute on \textit{Traffic} dataset.

\subsection{Experimental Setting}
All algorithms in the experiment are implemented in Java on CHAOS stream engine~\cite{gupta2009chaos}. All tests run on a computer equipped with Inter Xeon E5-2660 CPU (2.2GHz), 16G memory, and Windows Server 2012 operating system.
CHAOS platform supports multiple-dimensional data and count-based/time-based sliding window streams. The arrival rate of the streaming data also can be dynamic tuned in CHAOS engine. In our experiments the arrival rate is fixed as 500k tuples per second.

Our experimental study focuses on evaluating the effectiveness and efficiency of proposed evolving group and corresponding discovering algorithms. Therefore, we compare our proposed algorithm against the state-of-art $gathering$ pattern discovery algorithm~\cite{zheng2013discovery}\cite{zheng2014online} both on discovered patterns and utilized CPU time.
Specifically, we evaluate the effectiveness of our framework on \textit{Pedestrian} and \textit{Taxi}  datasets. We evaluate the efficiency and scalability of our online framework on large-scale \textit{Taxi} and \textit{Traffic} datasets. 
More specifically, we first do a pre-processing on the group files in \textit{Pedestrian} dataset and get a list of real groups for each time period, which is denoted as $Gr_{true}$. The list of groups detected by our method or competitor for each time period is denoted by $Gr_{test}$. 
So a group is stated as ``true positive" if it is both in $Gr_{test}$ and $Gr_{true}$. A group is stated as ``false positive'' if it is in $Gr_{test}$ but not in $Gr_{true}$. A group is stated as ``false negative" if it is not in $Gr_{test}$ but in $Gr_{true}$. \textit{Precision} and \textit{recall} are then calculated based on the above metrics. 

\subsection{Effectiveness}
\label{subsec.effect}

\subsubsection{Effectiveness on Pedestrian data}
\label{subsubsec:effonPedes}
First, we evaluate the effectiveness of our proposed evolving group pattern on $Pedestrian$ data compared against \textit{gathering} pattern \cite{zheng2013discovery}\cite{zheng2014online}. Figure~\ref{fig.prworkday} and Figure~\ref{fig.prweekend} show the comparison results of two patterns on workday and weekend, respectively. We set $MinPts$=2 and $Eps$=1.5 meters for DBSCAN, and set $w$=14, $k_c$=10, $m_c$=2, $k_p$=8, $m_p$=2 and $d$=5 meters for evolving group. Since moving groups in shopping center generally do not change for a short period, we set $m_g$=1 and $k_g$=14. Correspondingly, we set $k_c$=28, $m_c$=2, $k_p$=16, $m_p$=2, and $d$=5 meters for gathering pattern.
As we can see, the \textit{recall} of evolving group is much better than that of gathering pattern in all periods on both workday and weekend, while evolving group can also achieve similar \textit{precision} to gathering. This indicates that our evolving group can capture more actual groups compared to gathering. This is because people in a group (e.g., a family) may be separated briefly when they are shopping in the shopping center. 
As shown in Figure~\ref{fig.pworkday}, evolving group outperforms gathering pattern in term of \textit{precision} in all periods but morning on workday. 
Because people in the shopping center are generally relative sparse in the morning of workday, thus dense crowds are easily identified using gathering pattern. 
However, evolving group significantly outperforms gathering in term of \textit{recall} on workday morning (Figure~\ref{fig.rworkday}). This is because evolving group not only discovers the dense crowds but also captures the relaxed groups. 

\begin{figure}[htbp]
\centering
\subfigure[Precision]{
\label{fig.pworkday}%
\includegraphics[width=0.4\textwidth]{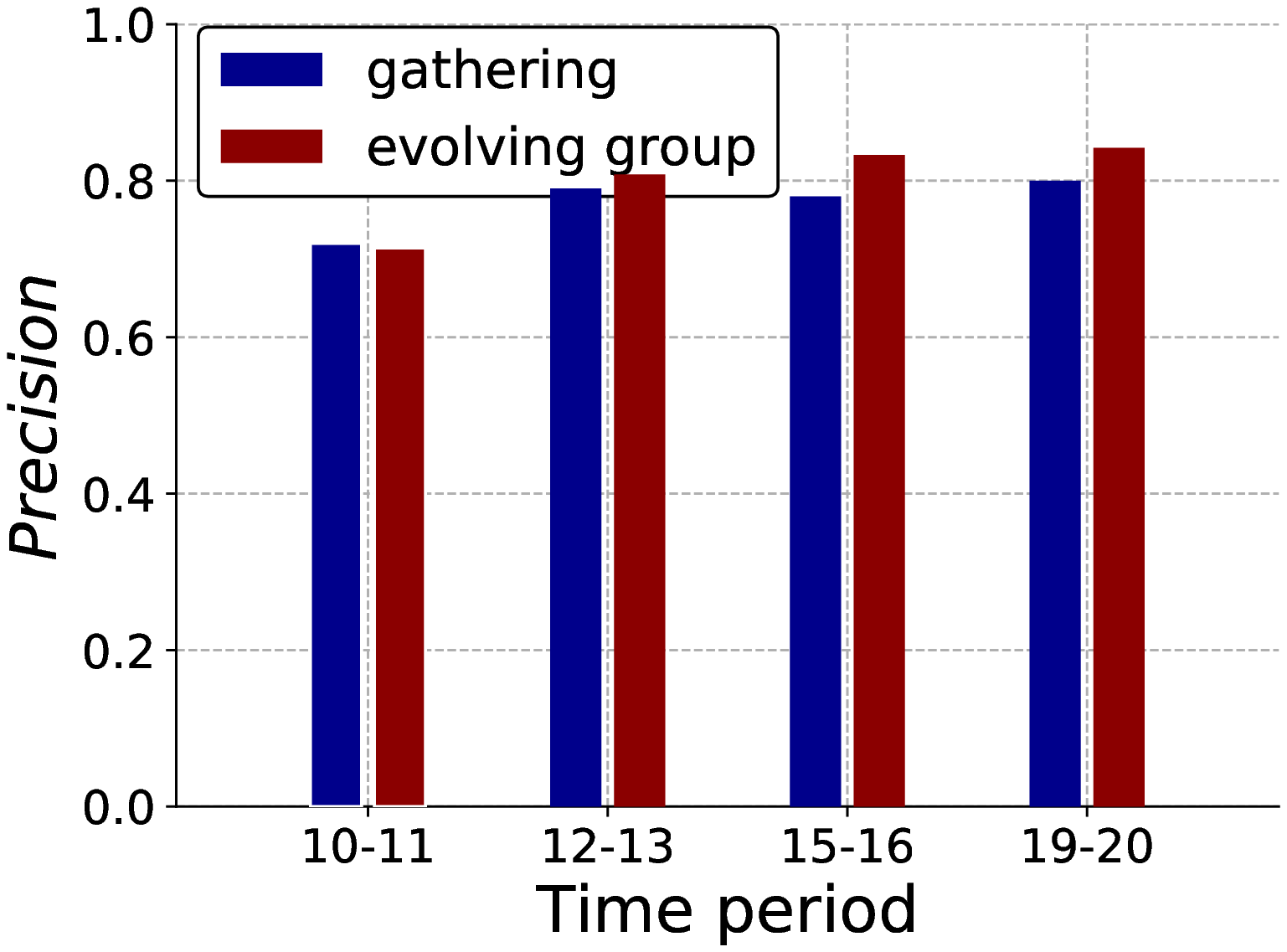}}
\hspace{5mm}
\subfigure[Recall]{
\label{fig.rworkday}%
\includegraphics[width=0.4\textwidth]{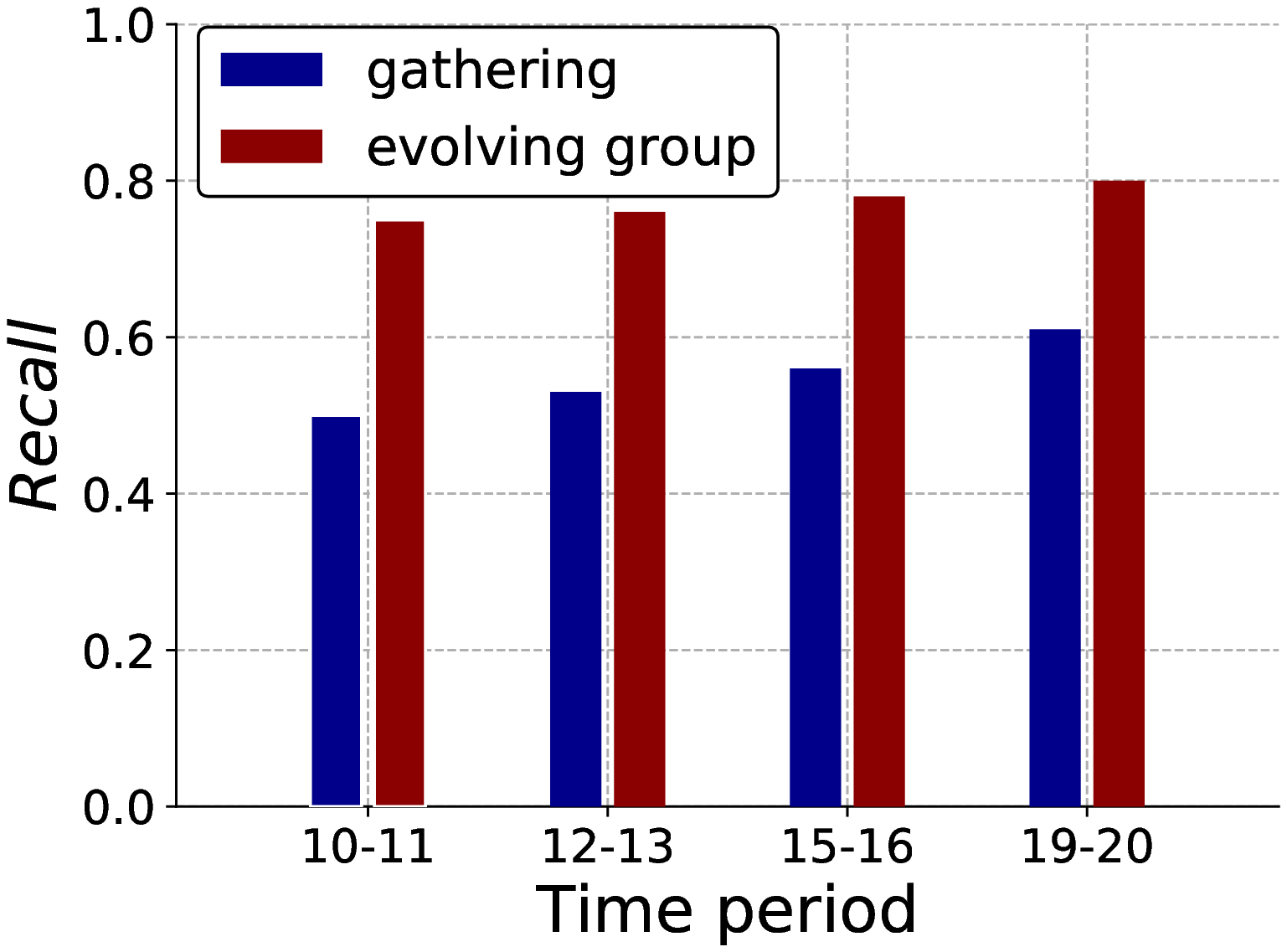}}
\caption{Performance comparison on workday}
\label{fig.prworkday}
\end{figure}

\begin{figure}[htbp]
\centering
\subfigure[Precision]{
\label{fig.pweekend}%
\includegraphics[width=0.4\textwidth]{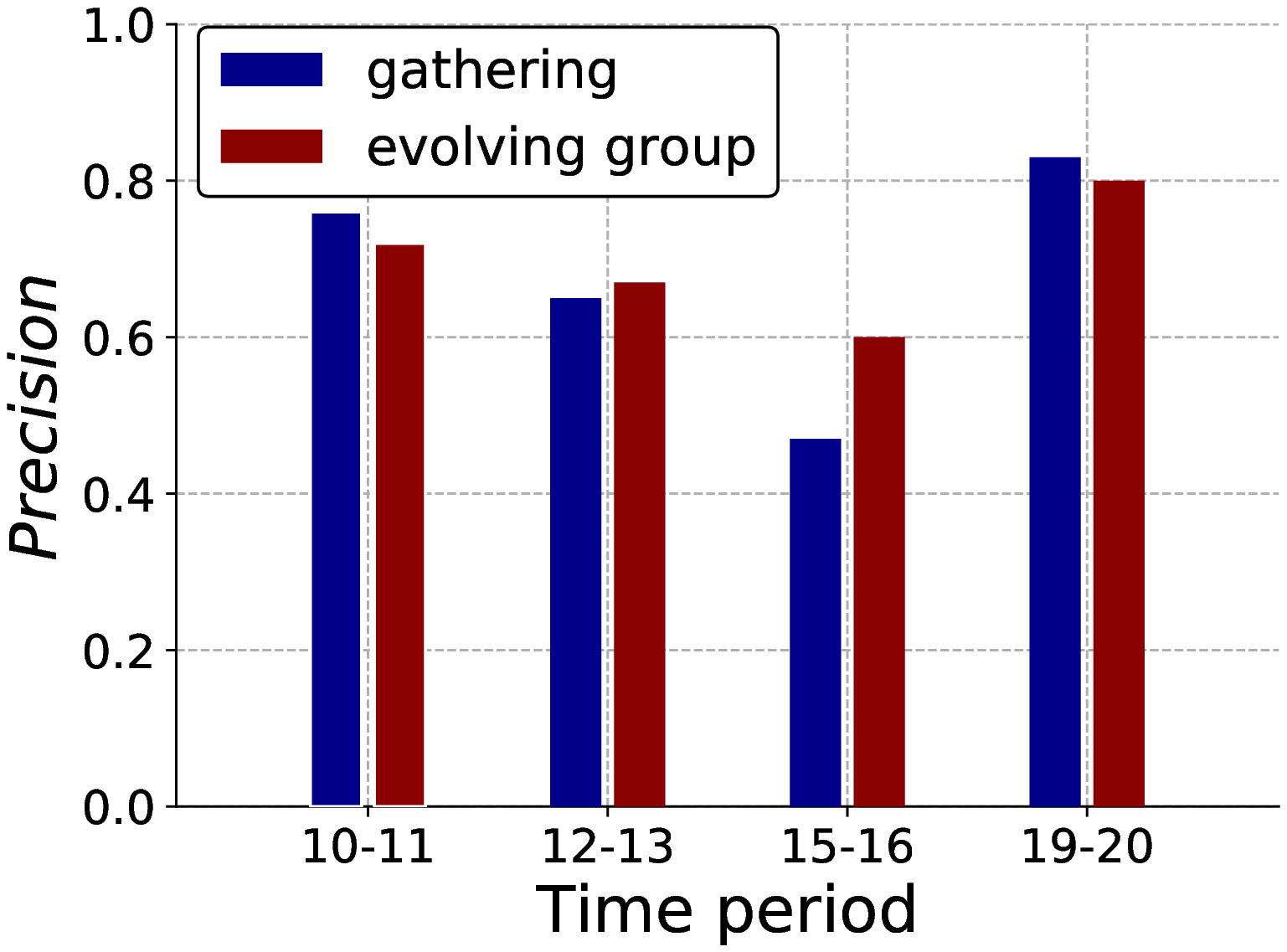}}
\hspace{5mm}
\subfigure[Recall]{
\label{fig.rweekend}%
\includegraphics[width=0.4\textwidth]{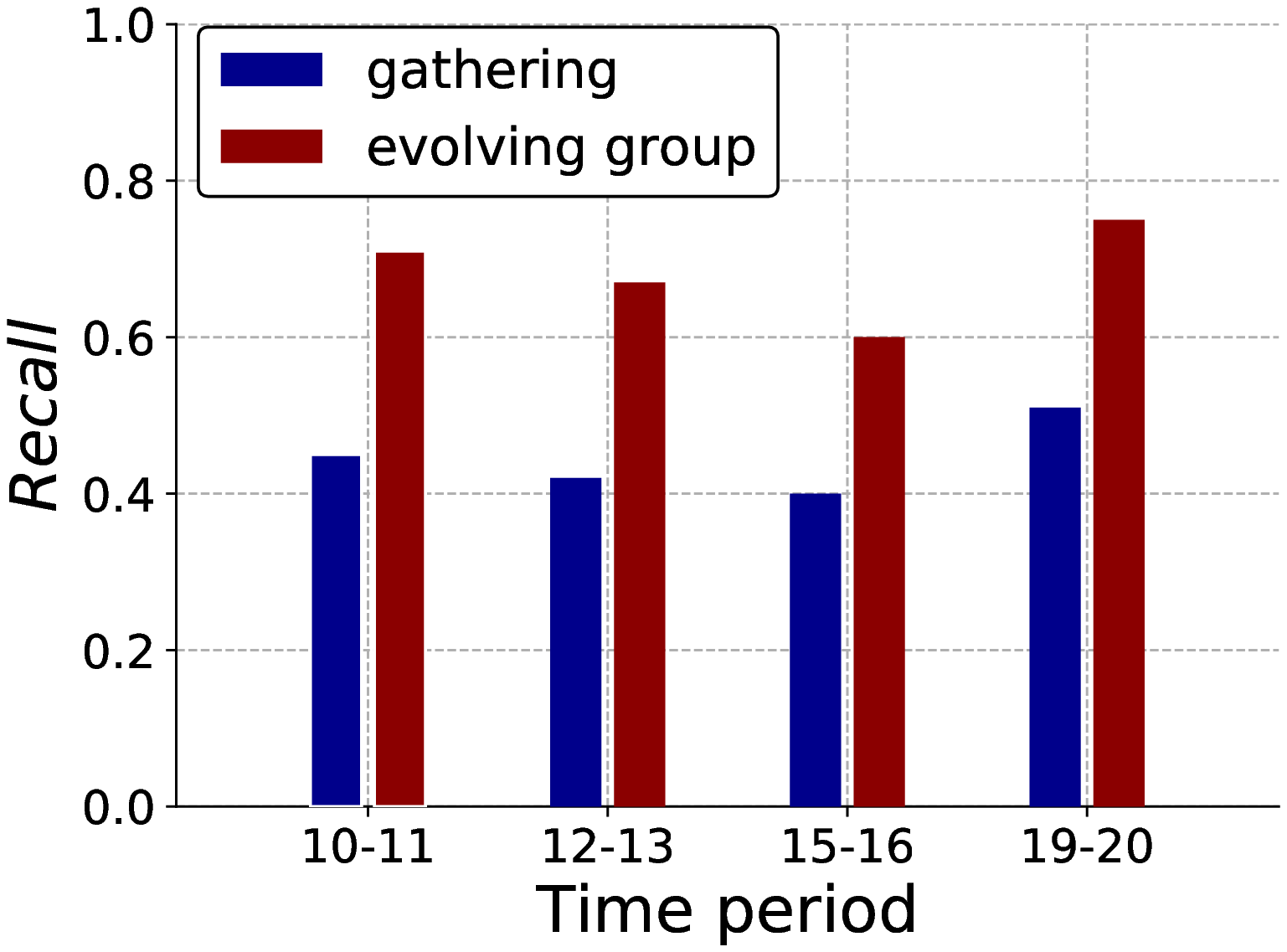}}
\caption{Performance comparison on weekend}
\label{fig.prweekend}
\end{figure}

From Figure~\ref{fig.prweekend}, we can see that our evolving group significantly outperforms gathering pattern in both precision and recall on weekend afternoon. This is because there are lots of people in the shopping center on weekend afternoon. Gathering pattern easily misses the actual relaxed groups in crowding scenes due to strict continuous time constraint, while evolving group just can capture such groups by ignoring some timestamps in the sliding window to fit the characteristics of groups in the real-world applications.

\subsubsection{Effectiveness on Taxi data}

We now evaluate the effectiveness of our proposed evolving group in case study of traffic condition on $Taxi$ data.
Intuitively, a serious traffic congestion is easily captured by a gathering pattern, because many vehicles aggregate in a dense cluster with slow speeds for a relatively long time. However, how a traffic jam is developed and formed? What is causal relationship between the contiguous short jams? What is the trend of a traffic congestion? Is the congestion becoming more serious? or being easing up? The gathering pattern is not applied to capture the evolving problems, while our proposed evolving group pattern exactly is to discover the evolving group events in dynamic trajectory streams.

In this experiment, we also divide the $Taxi$ data into two categories: workday and weekend.
We obtain the snapshot clusters at each timestamp by setting $MinPts$=$5$ and $Eps$=$300$ meters. Figure~\ref{fig.number} shows the average number of patterns discovered by our proposed algorithm on a single day with the settings of $w$=6, $k_c$=5, $m_c$=8, $k_p$=4, $m_p$=5, $m_g$=$0.7$, and $d$=$300$ meters (i.e., a group of 5 or more core members travelling at least 5 timestamps in a 6-minute sliding window). We select the closed evolving groups with $k_g$ $\geq$ 9 (i.e., an evolving group lasting for at least 9 consecutive windows). As comparison, we also search for the gathering patterns at the corresponding settings $k_c$=15, $m_c$=8, $k_p$=11, $m_p$=5 (i.e., a gathering of 5 or more participators travelling together for a period of at least 15 minutes).

In Figure~\ref{fig.workday}, we can find that the overall trend of evolving groups is consistent with the number of gathering patterns, which reflects the severe traffic congestion during the rush time on workday in Beijing. However, more evolving groups are captured in traffic streams compared with gathering pattern, especially during two peak times. This is because gathering only focuses on the serious traffic jams that last for a period of fixed consecutive time units, while our evolving group also tracks the short-lived aggregations of vehicles during non-consecutive time to monitor if they are becoming more serious or getting ease, except the long traffic congestions. 
Figure~\ref{fig.weekend} shows the number of discovered patterns on a weekend day. As we see, there are most traffic jams during morning peak time, and afternoon time is followed by. This is because, we learn that the weekend (February 2-3, 2008) approached the Spring Festival of China. Many companies arranged working days on that weekend. Therefore the results shows both characteristic of workday and weekend. However, compared to gathering, our evolving group also detects the traffic congestions during afternoon time, which is consistent with traffic conditions on the weekend afternoon before the Spring Festival in Beijing. Most of citizens go out to purchase the necessities, food and gifts for Festival in core areas of business street and shopping malls, or visit relatives and friends on weekend before the Spring Festival.

Next, we compare the average length of discovered patterns on a single workday and a weekend day. We select the closed evolving groups whose length $ k_g \geq$ 15 to make number of evolving groups be equal to the number of gatherings. From Figure~\ref{fig.length}, we can easily see that the average length of evolving group is larger than that of gathering at all time period, meaning that our group pattern can detect the crowding events earlier or track the trend of the events more time units.
In particular, the average length of evolving groups is much larger than that of gatherings during evening peak time, reaching at 15 minutes gap. This may be because that our framework captures the vehicle group events before the serious traffic jams are formed between 16:40 and 5:10 by observing the discovered patterns. During the afternoon on the weekend, our evolving group also senses the traffic congestions longer than gathering, reflecting the real non-smooth transportation condition.

\begin{figure}[htbp]
\centering
\subfigure[Workday]{
\label{fig.workday}%
\includegraphics[width=0.4\textwidth]{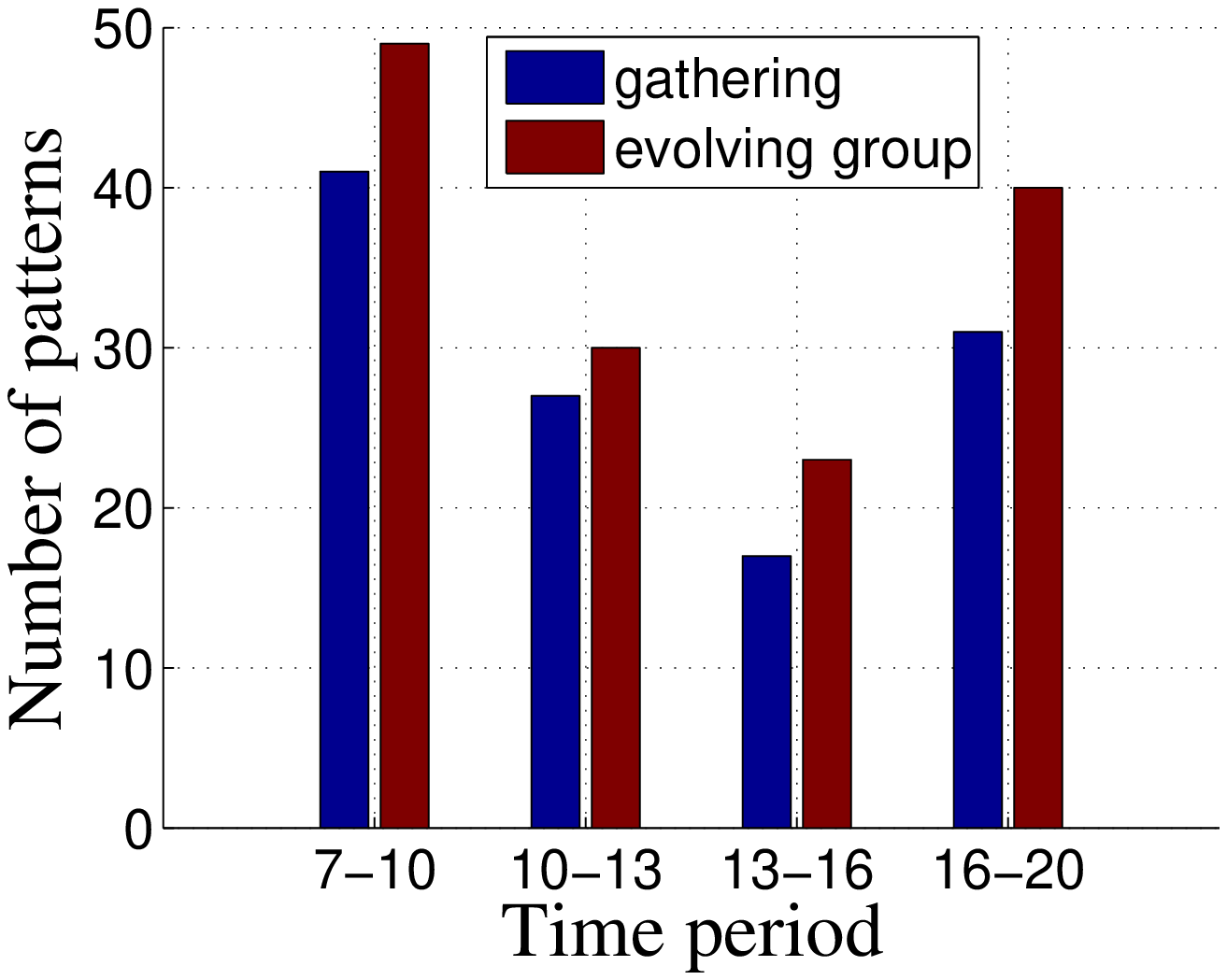}}
\hspace{5mm}
\subfigure[Weekend]{
\label{fig.weekend}%
\includegraphics[width=0.4\textwidth]{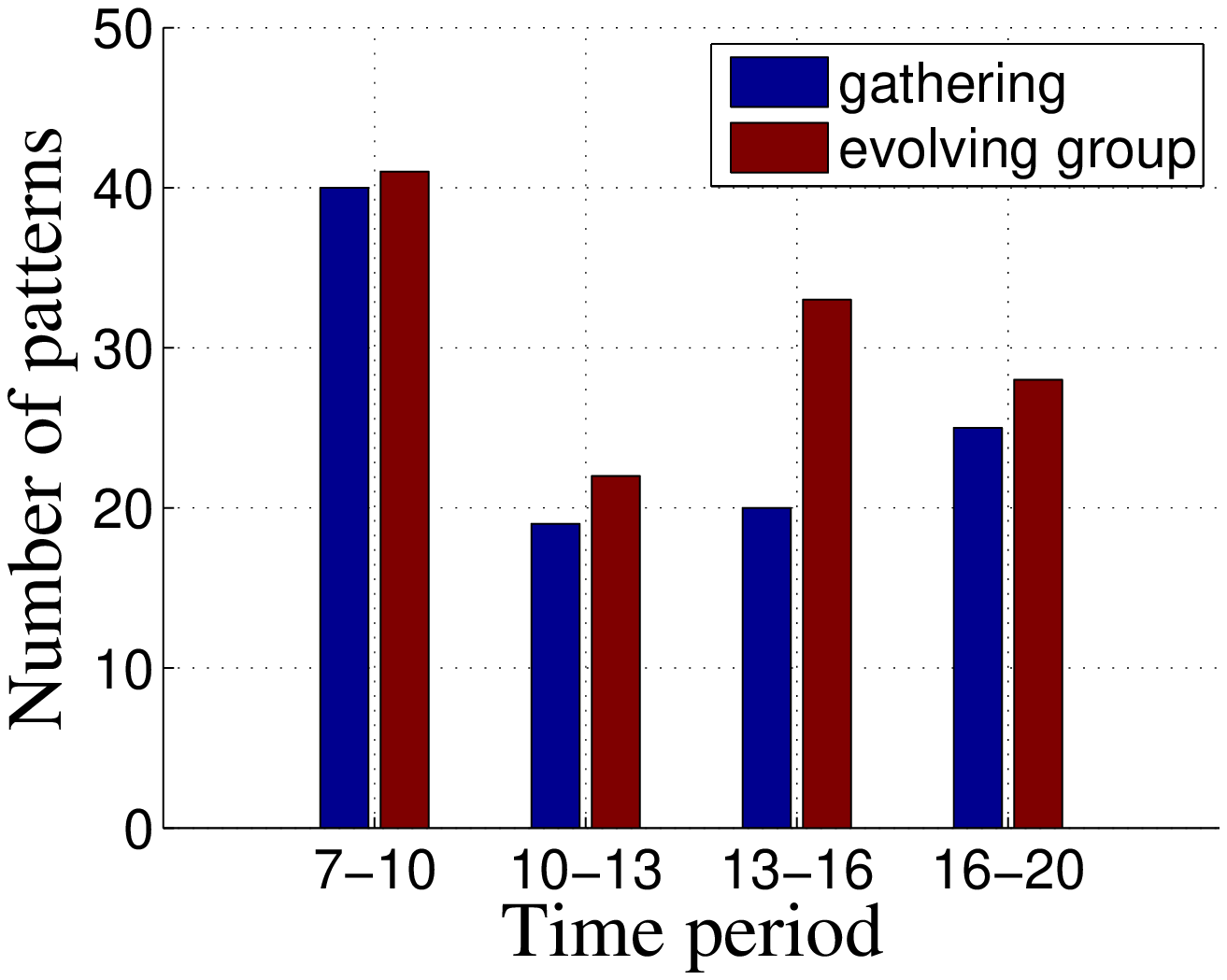}}
\caption{Number of discovered patterns}
\label{fig.number}
\end{figure}

\begin{figure}[htbp]
\begin{center}
\subfigure[Workday]{
\label{fig.workdaylen}%
\includegraphics[width=0.4\textwidth]{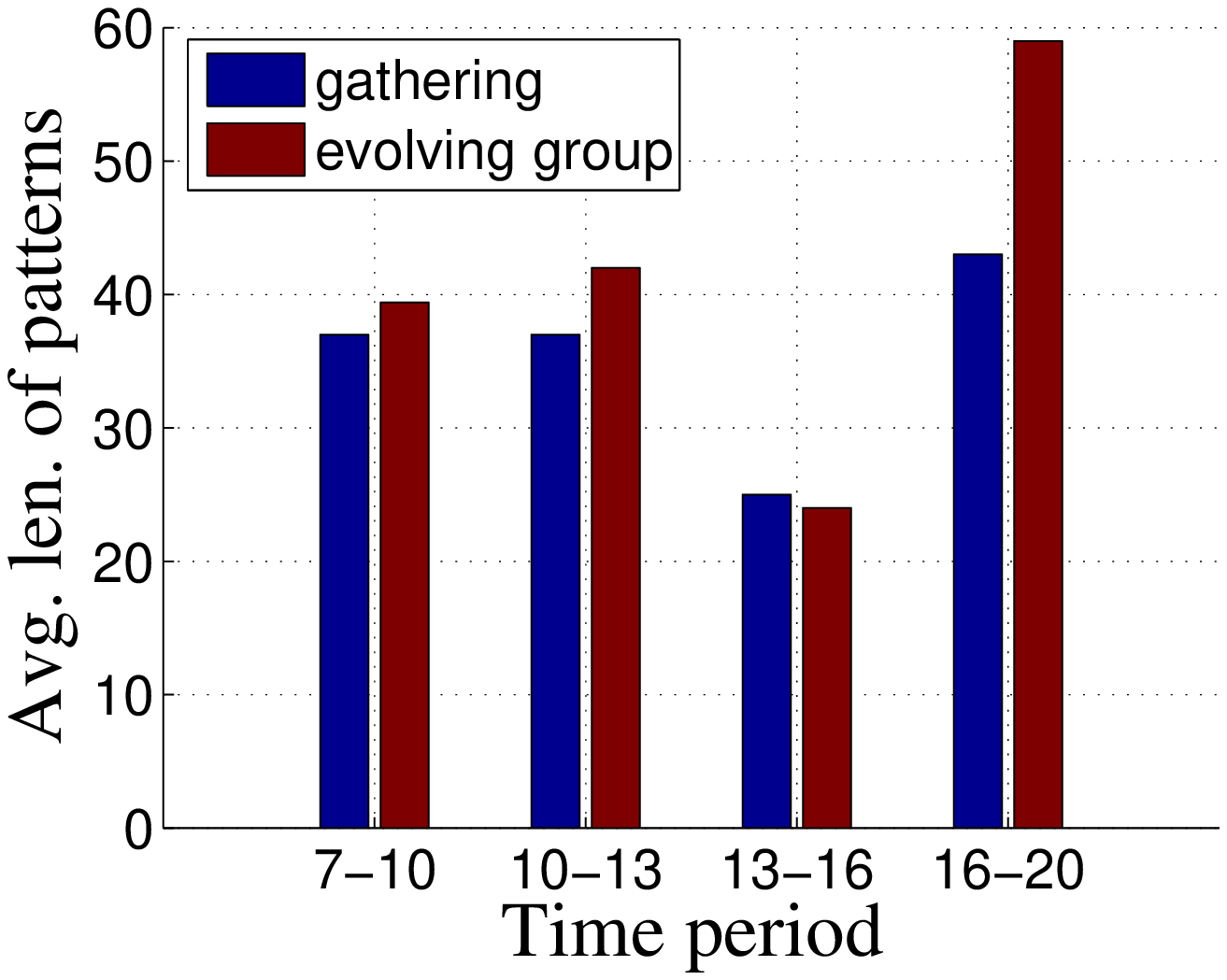}}
\hspace{5mm}
\subfigure[Weekend]{
\label{fig.weekendlen}%
\includegraphics[width=0.4\textwidth]{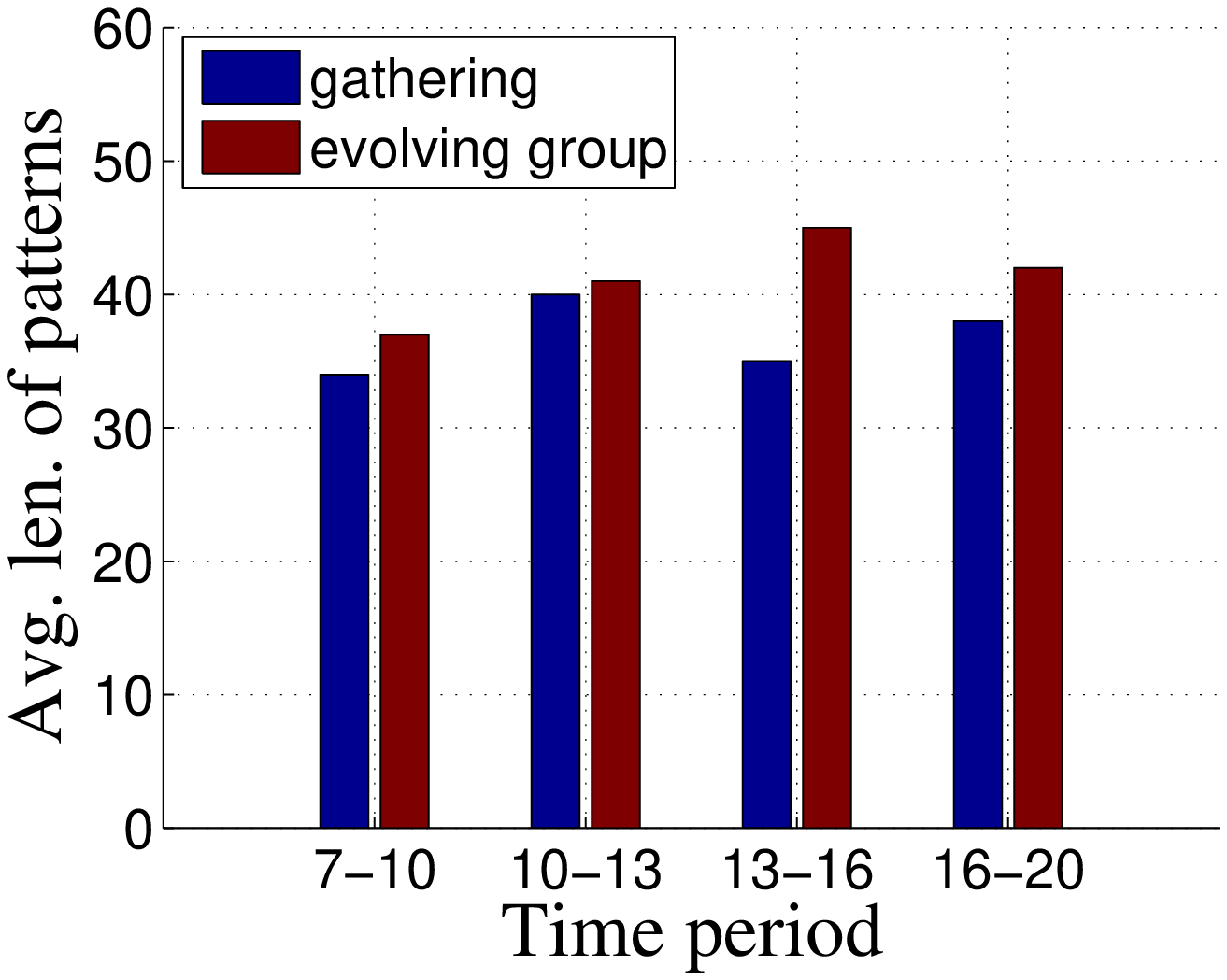}}
\caption{Average length of discovered patterns}
\label{fig.length}
\end{center}
\end{figure}


\begin{figure*}[tb]
\centering
\subfigure[case 1]{
\label{fig:casestudy1}
\includegraphics[width=0.95\textwidth]{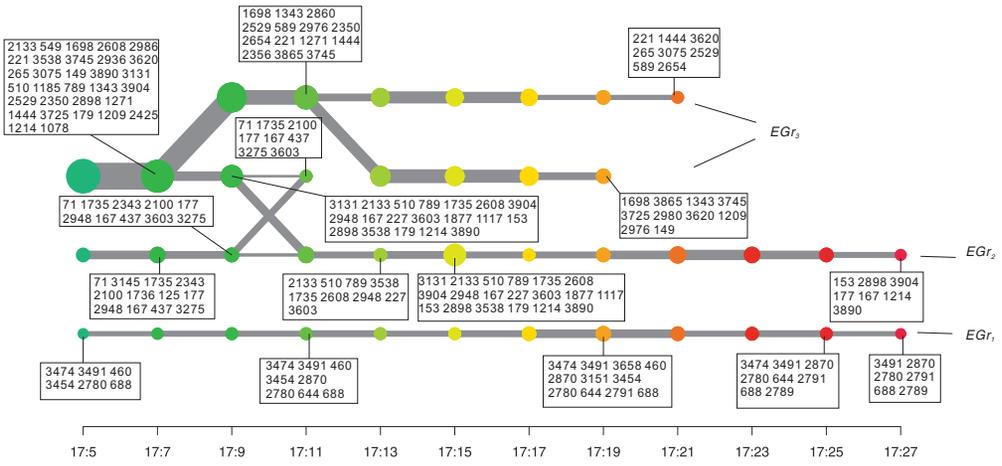}}
\subfigure[case 2]{
\label{fig:casestudy2}
\includegraphics[width=0.95\textwidth]{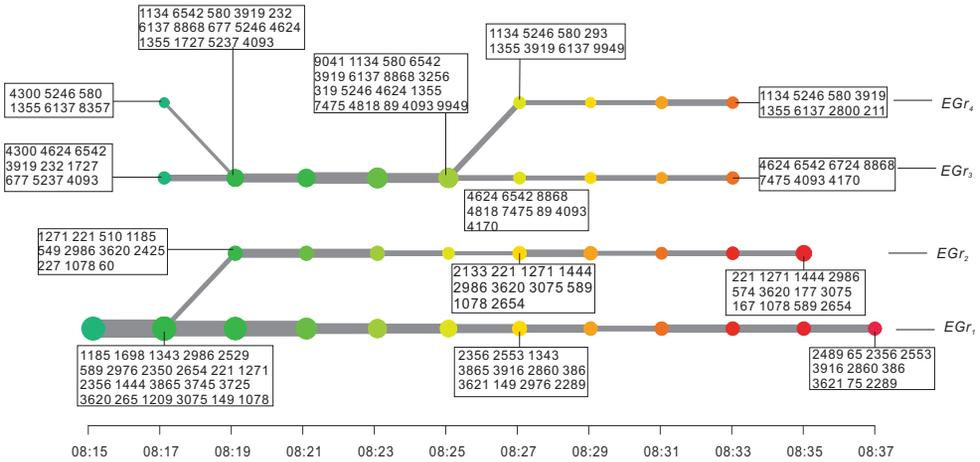}}
\caption{Case studies on effectiveness}
\label{fig:casestudy}
\end{figure*}

\subsubsection{Case studies}

We employ a visualized method with two cases to demonstrate the effectiveness of our evolving group. As shown in Fig. \ref{fig:casestudy}, each circle represents a group in sliding window. For two connected groups by a solid line, the latter is evolved from the former, and the weight of line indicates the number of common members. We also use the size of circle to denote the number of objects in the group, and the color of circle to denote the life time of the evolving group, gradually changing from green to red. We also plot out the IDs of objects in the rectangle for some significant groups. It is easy to see the advantages of evolving groups from Fig. \ref{fig:casestudy1} and \ref{fig:casestudy2}. Even if a gathering pattern has the same length as an evolving group, it can not contain such rich information as an evolving group reveals.

\textbf{Case 1:}
Fig. \ref{fig:casestudy1} shows several evolving groups during the period of 17:05-17:27 on Feb 4. We simply mark out 3 representative evolving groups as shown in Fig. \ref{fig:casestudy1}.
Obviously, $EGr_1$ is independent of other evolving groups, however, it evolves continuously over time, for example, the evolving group grows to 11 core members at clock 17:19, while shrinks into 6 participators at 17:27.
$EGr_3$ represents a serious traffic jam for 6 minutes at beginning of the case. However, the traffic jam gets alleviated by splitting into two smaller groups at 17:13, and then the two evolving groups end at about 17:19 and 17:21 respectively, meaning that transportation condition becomes much smoother.
Another $EGr_2$ interacts with $EGr_3$ during [17:09, 17:11]. We observe that $EGr_2$ suddenly becomes much bigger at time 17:15, and then be much smaller at next window. This may be caused by an emergency, such as a traffic accident or an emergency repair.

\textbf{Case 2:}
In Fig. \ref{fig:casestudy2}, four evolving groups during 8:15-8:37 on Feb 4 are shown in this case.
$EGr_1$ shows the development process of a continuous traffic congestion from very serious towards somewhat light. However, we can see that the participators at time 8:17, 8:27 and 8:37 also change significantly over time, but our model also captures the traffic jam using gradual evolution in sliding window.
$EGr_3$ and $EGr_4$ share the first half part, which demonstrates the forming process of a serious traffic jam vividly. This can be revealed by that most participators keep evolving into next group continuously from 8:17 to 8:25. The cause of the phenomenon that the jam is separated into two groups at time 8:27 may be an efficient shunting strategy. Moreover, we can get that the two vehicle teams are scattered gradually from 8:28 to 8:35 based on the evolving groups shown in Fig. \ref{fig:casestudy2}.

\subsection{Efficiency}
Next, we compare the performance of our DEG method with the discovery algorithm of gathering pattern in \cite{zheng2014online}. We denote their crowd detection and closed gathering discovery in gathering pattern \cite{zheng2014online} as \textbf{G-crowd} and \textbf{TAD*} respectively.
In particular, we measure the running time of each window in different parameter settings. The results are averaged over ten thousand windows. Each window slides by one minute. 
Since $k_c$ and $m_c$ mainly affect the number of crowds, we only measure the running time of closed crowds discovery with respect to $k_c$ and $m_c$ using two pruning methods: a) -prune, our cluster pruning strategy; b) -grid, grid-based indexing used in~\cite{zheng2013discovery}\cite{zheng2014online}.

\subsubsection{Running time w.r.t. thresholds $k_c$ and $k_p$}

\begin{figure}[htbp]
\centering
\subfigure[Threshold $k_c$]{\label{fig:kc}%
\includegraphics[width=0.4\textwidth]{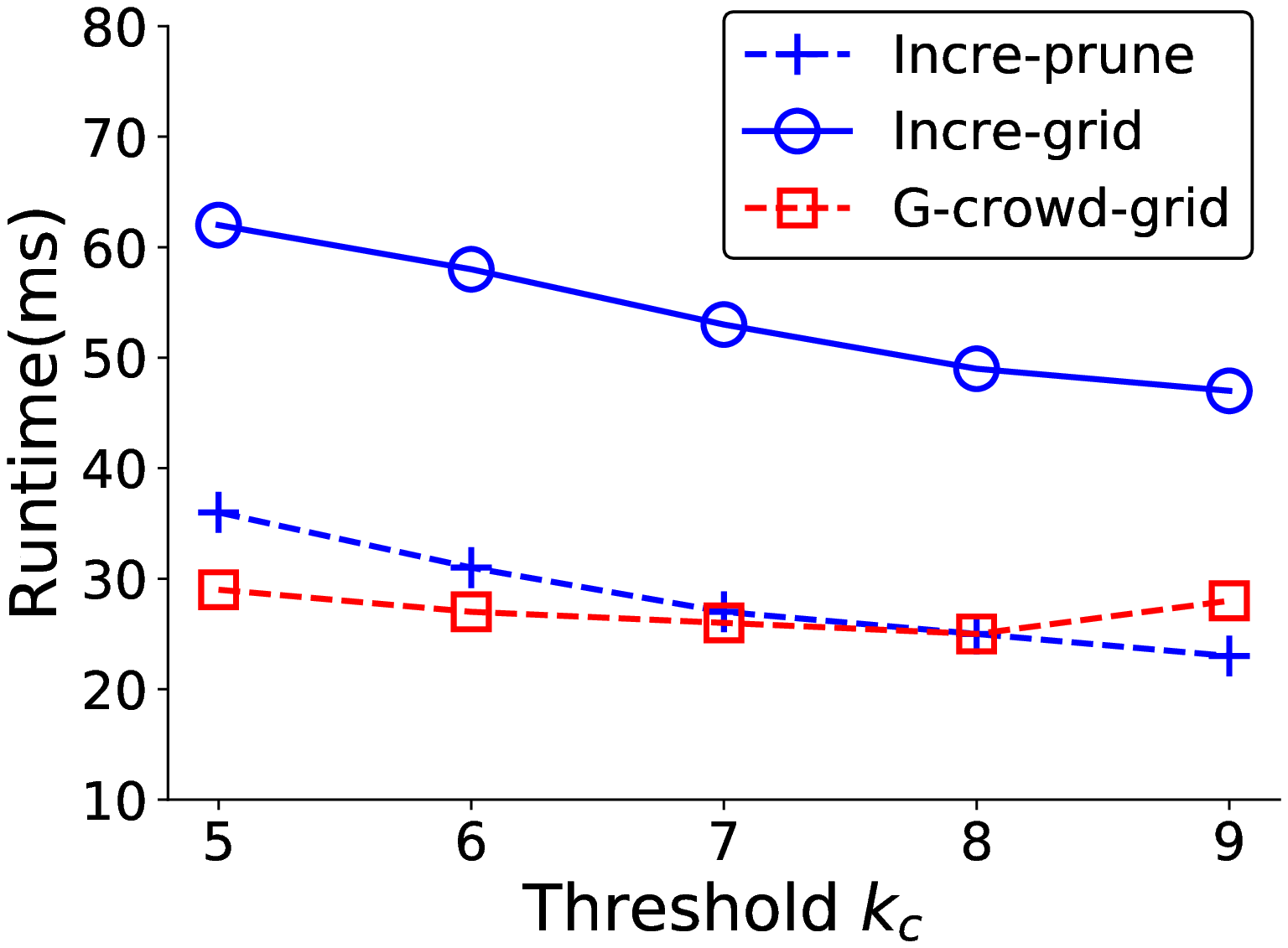}}
\hspace{5mm}
\subfigure[Threshold $k_p$]{\label{fig:kp}%
\includegraphics[width=0.4\textwidth]{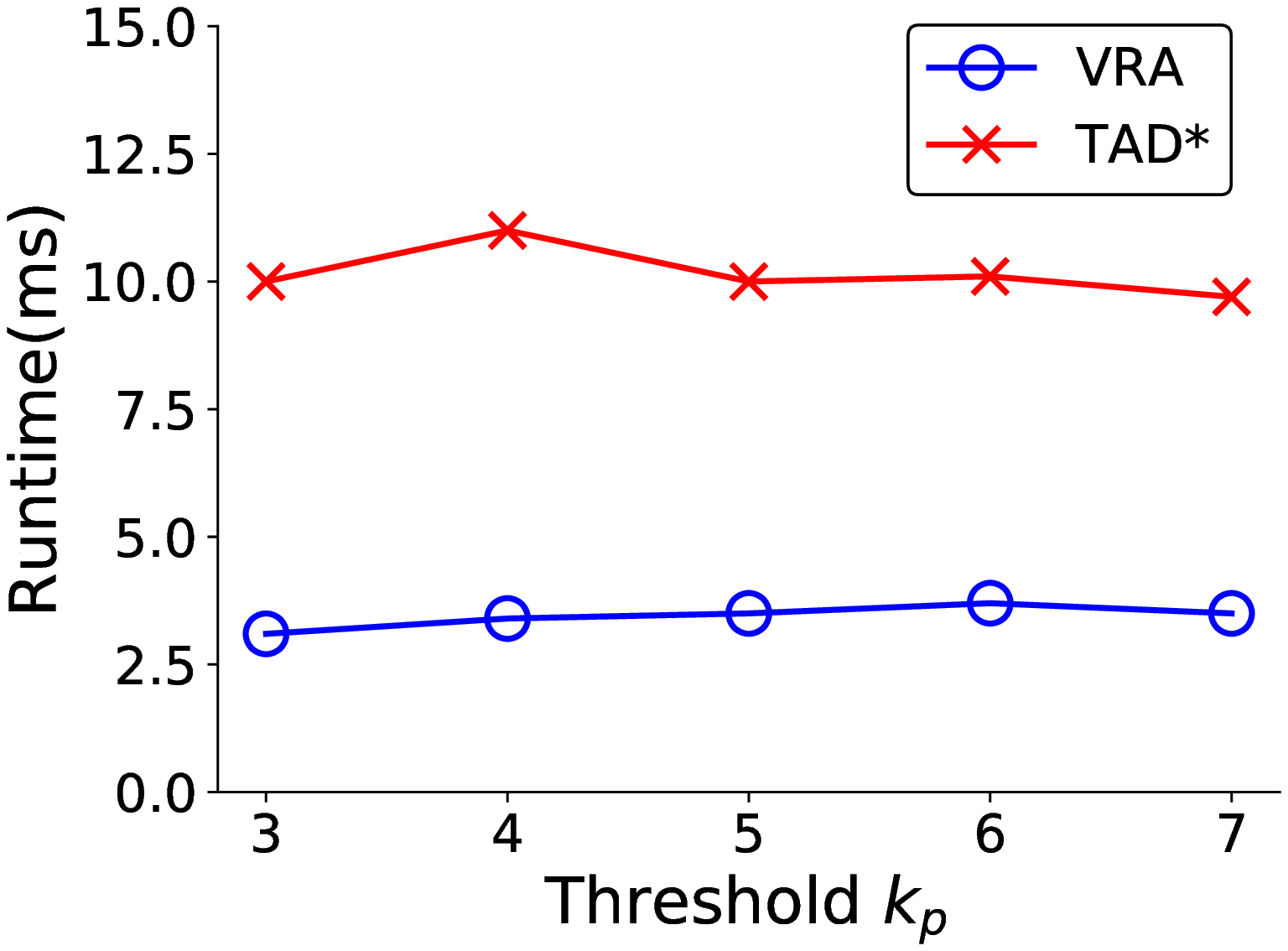}}
\caption{Running time w.r.t. thresholds $k_c$ and $k_p$}
\label{fig:kckp}
\end{figure}

We first evaluate the impacts of thresholds $k_c$ and $k_p$ on the performance of two methods when fix $w$=10 $m_c$=8, $m_p$=5, $d$=300 meter, $|O_{DB}|$=10,000, and $m_g$=0.7. Figure~\ref{fig:kc} shows the results of running time with respect to $k_c$ when $k_p$=5. Figure~\ref{fig:kp} shows the results of running time with respect to $k_p$ when $k_c$=9.

As shown in Figure~\ref{fig:kc}, the CPU time of our Incre method increases as $k_c$ decreases.  This is because $k_c$ in our evolving group definition would affect the number of closed crowd candidates stored in $CanSet$. The potential crowds whose length is not less than $k_c$ would certainly be maintained, resulting in more detection time of crowd candidates for smaller $k_c$. This is consistent with the above complexity analysis in Sec.~\ref{subsec.complex}.
$k_c$ in gathering pattern only affects the number of detected gatherings, thus has no impact on the detection time. 
However, Incre with our proposed cluster pruning strategy can achieve similar performance to G-crowd, and outperforms G-crowd when $k_c$=9 in context of $w$=10.

From Figure~\ref{fig:kp}, we can see that VRA is superior to TAD* in term of CPU time. This is because we reuse the detection process of participators in finding closed aggregation as window slides, although we maintain much more closed crowd candidates. In addition, $k_p$ has less impact on running time of both VRA and TAD* methods. This is because $k_p$ only affects the number of participators, but no impact on the time cost of detecting them.

\subsubsection{Running time w.r.t. thresholds $m_c$ and $m_p$}

\begin{figure}[htbp]
\centering
\subfigure[Threshold $m_c$]{\label{fig:mc}%
\includegraphics[width=0.4\textwidth]{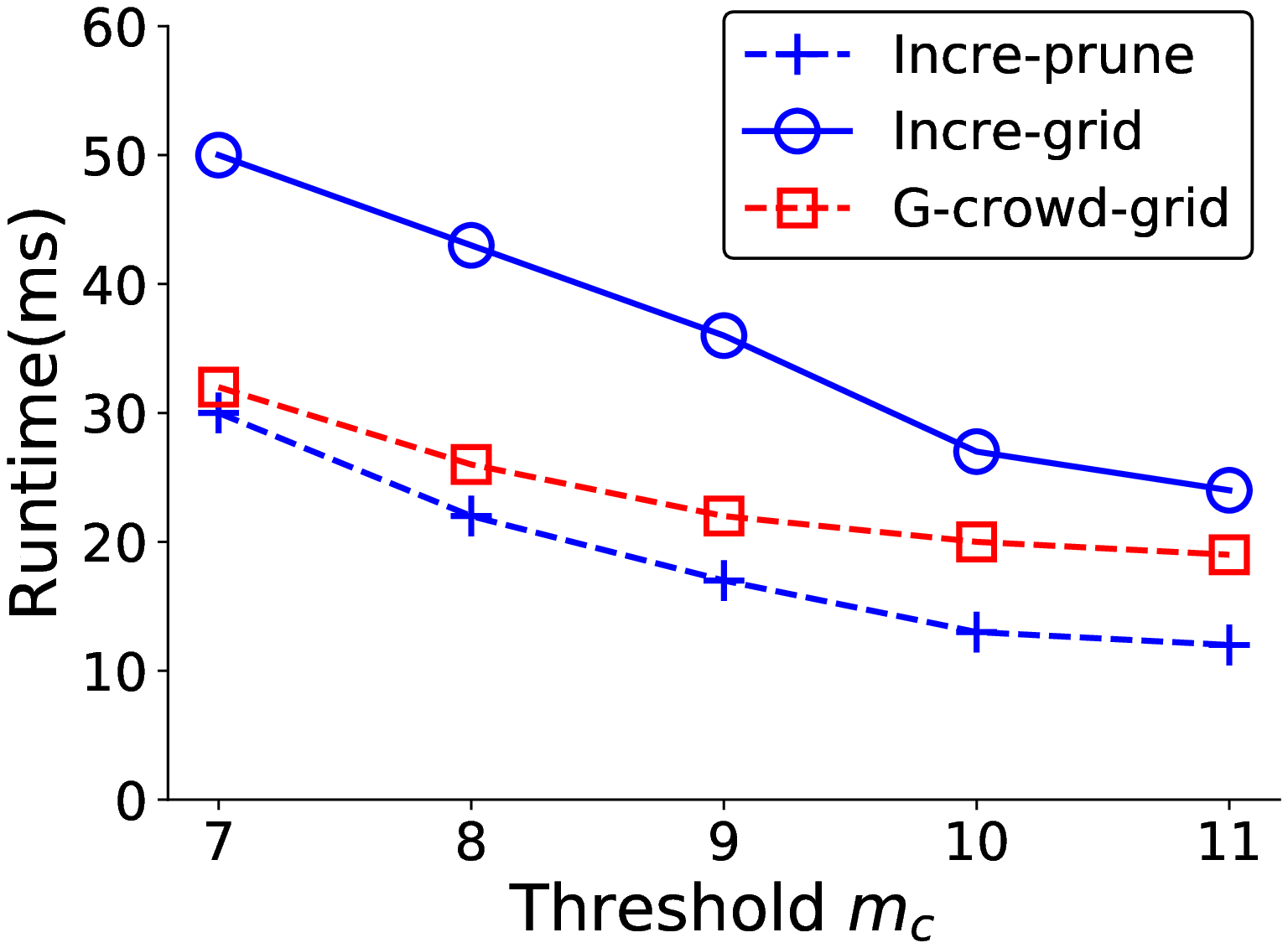}}
\hspace{5mm}
\subfigure[Threshold $m_p$]{\label{fig:mp}%
\includegraphics[width=0.4\textwidth]{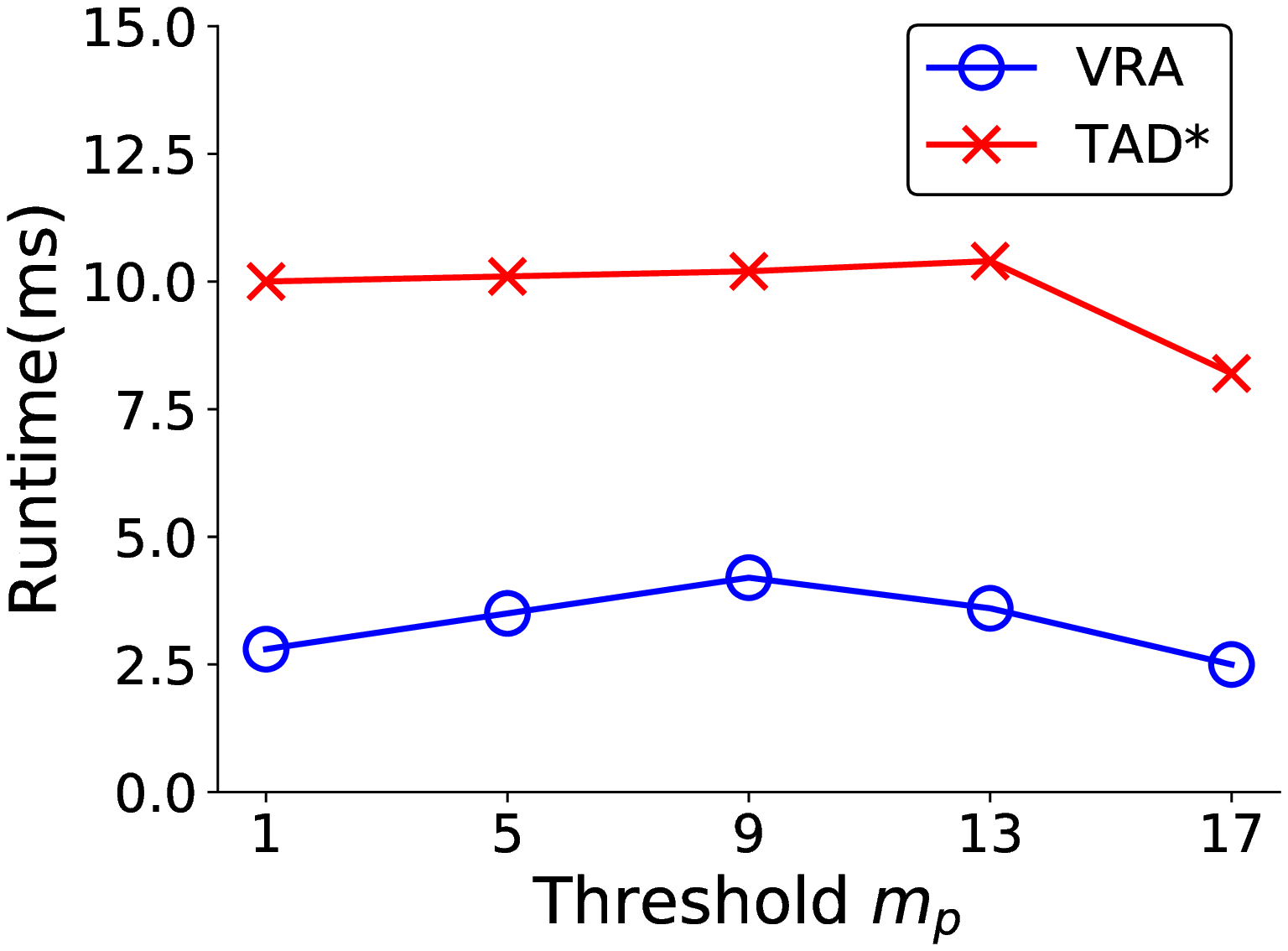}}
\caption{Running time w.r.t. thresholds $m_c$ and $m_p$}
\end{figure}

Next, we evaluate the performance of two methods w.r.t threshold $m_c$ and $m_p$. 
We fix $k_c$=7 ($w$=8), $k_p$=5, $d$=300 meter, $|O_{DB}|$=10,000, and $m_g$=0.7.
In Figure \ref{fig:mc}, we vary $m_c$ from 7 to 11 at fixed $m_p$=$5$. In Figure \ref{fig:mp}, we vary $m_p$ from 1 to 17 at fixed $m_c$=8.
As shown in Figure \ref{fig:mc}, we can see that Incre-prune outperforms G-crowd-grid and Incre-grid methods. This may be because 1) we use two optimization principles to reduce the number of crowd candidates and prune the unnecessary validating with new clusters; 2) our cluster pruning efficiently filters the long-distance clusters as well as short-distance clusters, while grid suffers expensive cost in indexing clusters in high-speed streaming data.
Since $m_c$ affects the number of clusters that satisfy the dense group at each timestamp, the larger $m_c$ we choose, the less number of clusters satisfying this threshold. Therefore, the running time of all algorithms decreases as $m_c$ increases.
VRA again exhibits much better performance than TAD* in term of CPU time in Figure \ref{fig:mp}. 
The reason is same as the above explained.
Moreover, $m_p$ would affect the number of invalid clusters in crowd. As $m_p$ increases, more clusters are invalid, which causes process of closed aggregation to terminate more quickly in each window for our framework.

\subsubsection{Running time w.r.t number of objects and hausdorff distance}

Next, we compare total running time of our DEG method against gathering pattern by varying the number of moving objects and Hausdorff distance respectively when fix $w$=8, $k_c$= 7, $m_c$ = 8, $k_p$ = 5, $m_p$ = 5, $m_g$=$0.7$.
Figure~\ref{fig:objects} shows the results at fixed $d$=$300$ meters.
Our algorithm is superior to the discovery framework of gathering pattern (with grid indexing) in detecting closed crowds and aggregations, shown as Incre+VRA (with pruning strategy).
Moreover, our overall framework DEG also outperforms the gathering pattern in term of total running time, although our evolving group can capture more interesting patterns by considering crowds across non-consecutive time points.
In particular, our framework reaches 31ms per window when $|O_{DB}|$=$10,000$, saving 16\% CPU time compared to the gathering pattern.
Figure~\ref{fig:distance} shows the results w.r.t Hausdorff distance at fixed $|O_{DB}|$=$10,000$. As expected, the running time of both algorithms increases as the distance increases. This is because the search space between new clusters and ending clusters increases although both employ pruning methods, further causing the number of crowd candidates increases at each window.

\begin{figure}[htbp]
\centering
\subfigure[Number of objects]{\label{fig:objects}%
\includegraphics[width=0.4\textwidth]{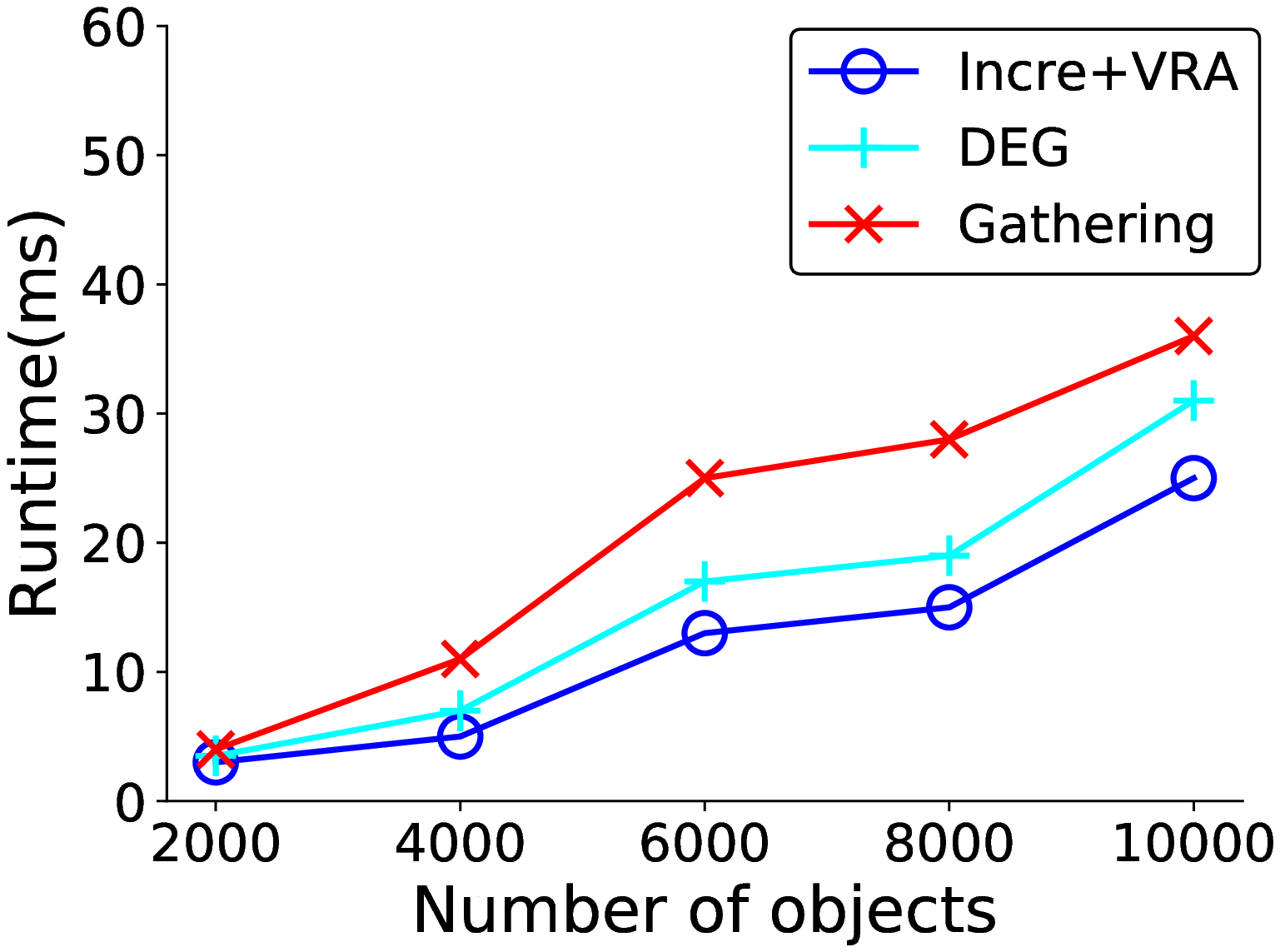}}
\hspace{5mm}
\subfigure[Hausdorff distance]{\label{fig:distance}%
\includegraphics[width=0.4\textwidth]{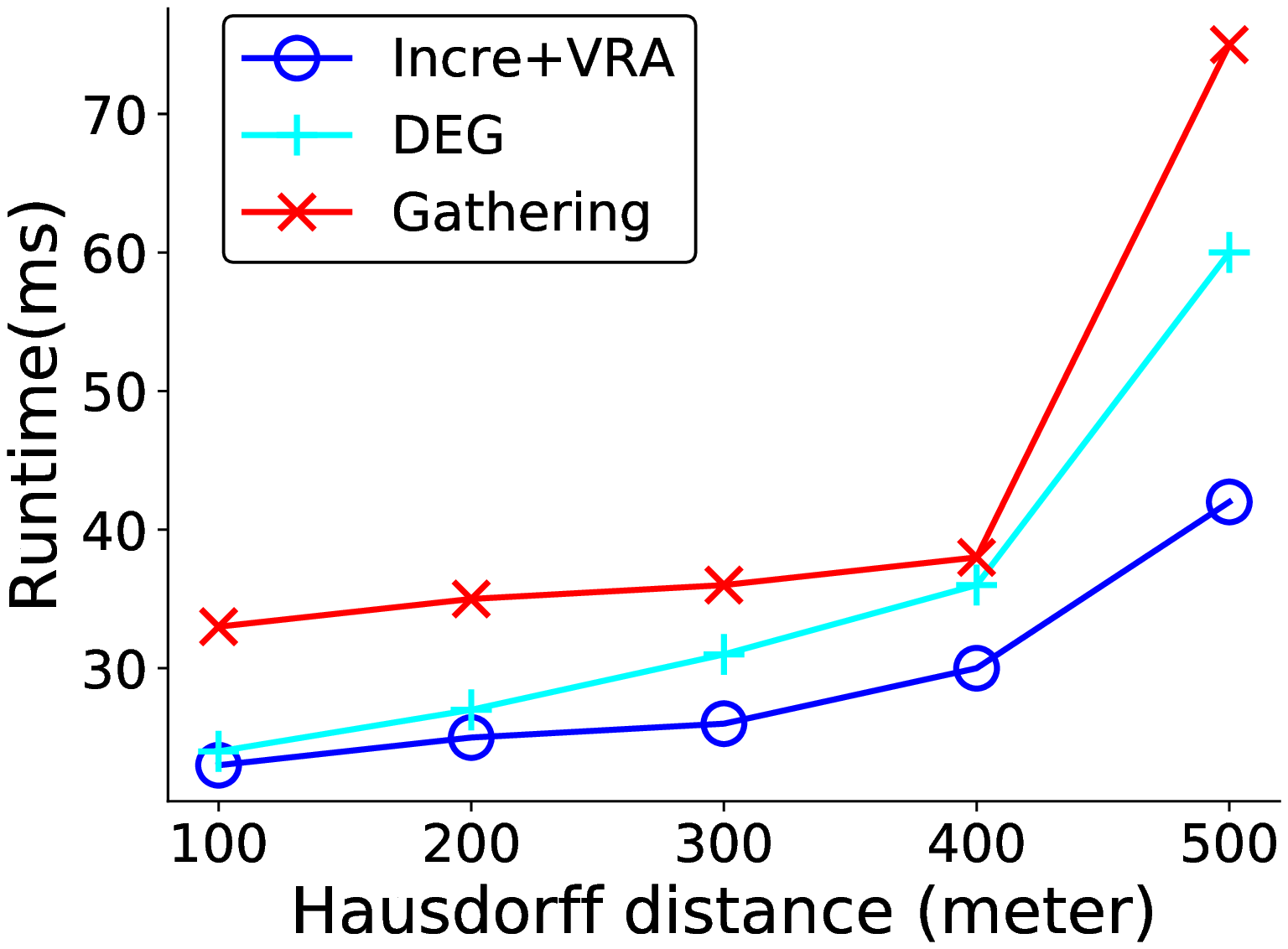}}
\caption{Running time w.r.t. number of objects and Hausdorff distance}
\end{figure}


\subsection{Scalability}
\label{subsec.scale}
Finally, we evaluate the scalability of our parallel MTOD framework on the large-scale \textit{Traffic} data with respect to the number of threads and the number of objects (vehicles). 
We fix $w$=8, $k_c$=7, $m_c$=10, $k_p$=$5$, $m_p$=10, $d$=$1,500$ meters, $m_g$=$0.7$, and $k_g$=10.
The results are shown in Figure~\ref{fig:scale}, \textit{Random} denotes multi-threading method using random distribution of new clusters. 
Figure~\ref{fig:threads} shows the results of all methods on all vehicles (around 180K) on the road network in \textit{Traffic} data by varying the number of threads. Figure~\ref{fig:scale.objects} shows the results of all methods with respect to the number of vehicles when the number of threads is fixed to 16, i.e., $|O_{DB}|$ varies from 60K to 180K.
As shown in Figure~\ref{fig:threads}, our MTOD framework consistently outperforms our proposed serial DEG algorithm and Random. In particular, MTOD achieves 10 speed up over our DEG algorithm when the number of threads equals to 8, while Random only improves up to 6 times faster than DEG algorithm. This is because our MTOD uses sector-based partition to distribute the new clusters into each thread, which quickly prunes the closed crowd candidates in non-adjacent sectors without computing. Therefore, our MTOD successfully achieves a super-linear speed-up ratio.
From Figure~\ref{fig:scale.objects}, we can see that our MTOD framework shows better scalability than our serial DEG method and is superior to Random with respect to the number of objects. More specifically, MTOD framework significantly outperforms DEG by on average 10.6 times in runtime on all tested cases, whereas Random only improves on average 6.8 times over DEG. As explained above, randomly distributing only achieves parallel discovery of closed aggregations for each partition, but it still checks all closed crowd candidates in $CanSet$ for each new clusters, while our MTOD directly prunes more crowd candidates before applying the long-distance pruning.

\begin{figure}[htbp]
\centering
\subfigure[Number of threads]{\label{fig:threads}%
\includegraphics[width=0.4\textwidth]{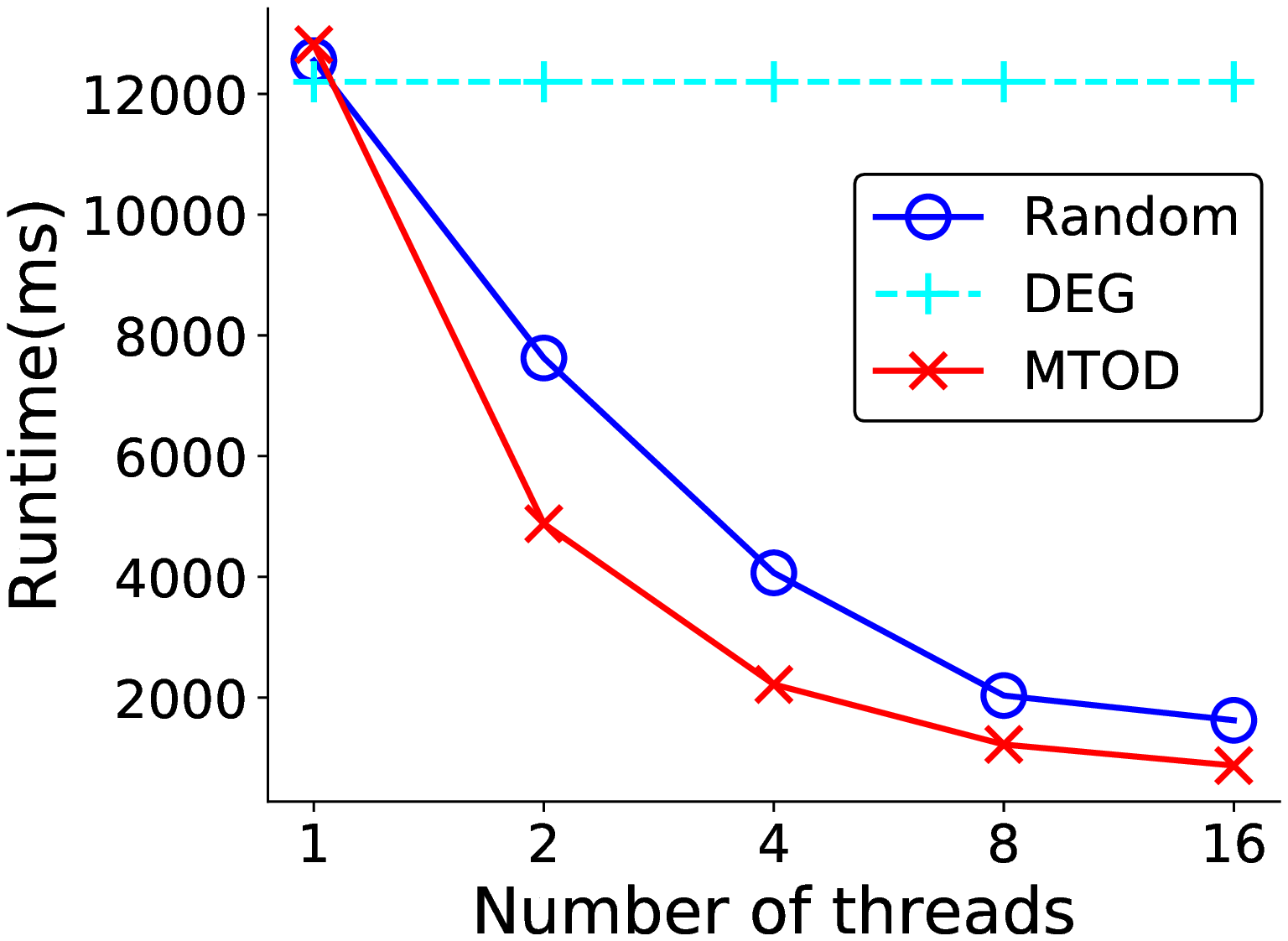}}
\hspace{5mm}
\subfigure[Number of objects]{\label{fig:scale.objects}%
\includegraphics[width=0.4\textwidth]{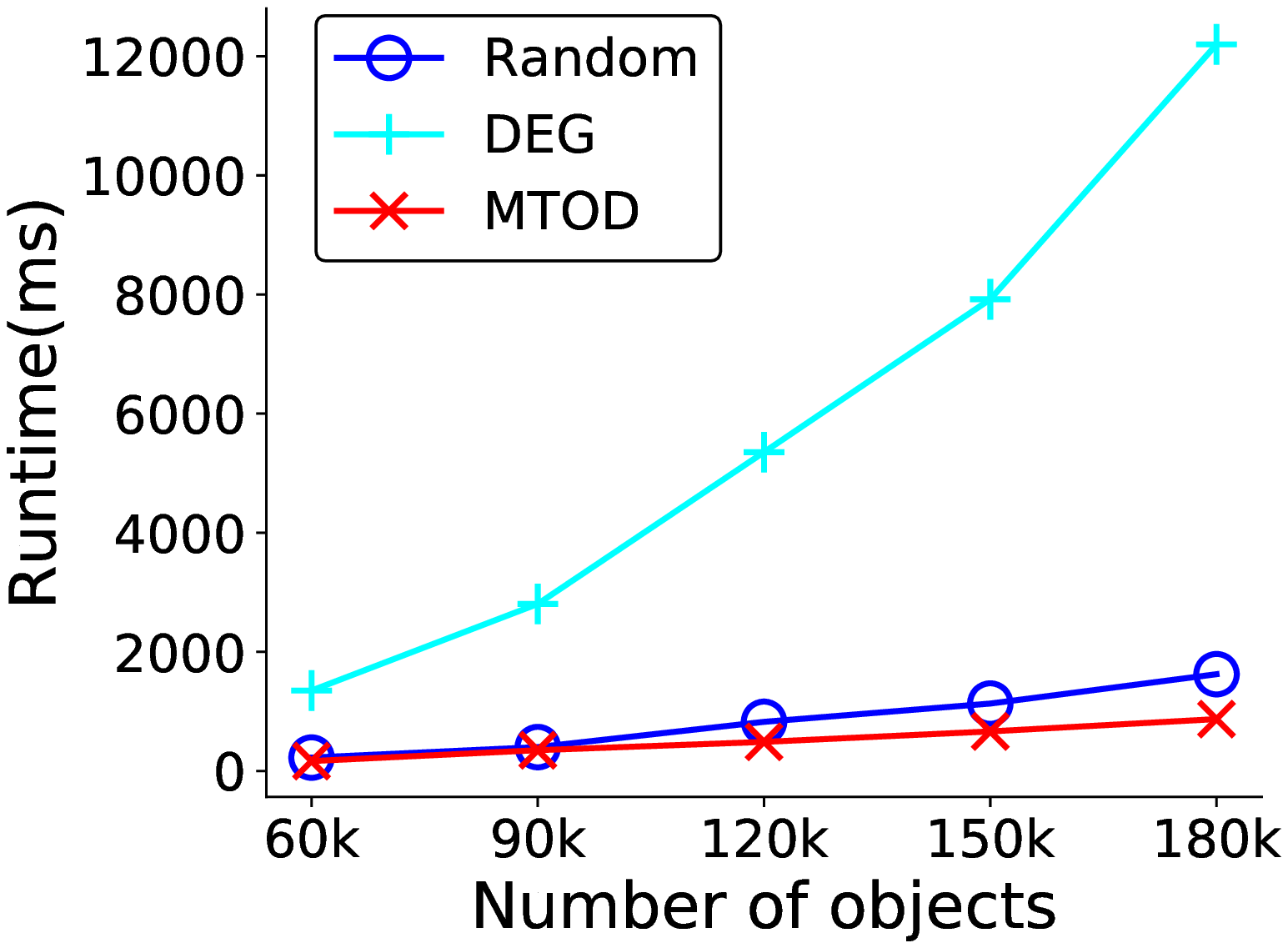}}
\caption{Efficiency comparison of parallel framework}
\label{fig:scale}
\end{figure}
\section{Conclusion}
\label{sec.con}
In this work we focus on the detection of evolving group patterns over massive-scale trajectory streams. After analyzing the requirements of stream trajectory monitoring applications, we first propose the novel concept of evolving group to capture the variety of group events and their evolving process in trajectory streams. Moreover, we design an online discovery algorithm of evolving group, which contains three phases incorporating a series of optimization principles to reduce computation cost. Furthermore, we extend a multi-threading based parallel discovery framework to scale to huge-scale trajectory streams. At last we evaluate the effectiveness and efficiency of our method compared against the state-of-the-art on three real wold large-scale datasets.


\section*{Acknowledgment}
This work is partially supported by the National Natural Science Foundation of China (nos. 61773331, 61703360, 61403328 and 61502410).

\bibliographystyle{ACM-Reference-Format}
\bibliography{group}


\begin{thebibliography}{38}


\ifx \showCODEN    \undefined \def \showCODEN     #1{\unskip}     \fi
\ifx \showDOI      \undefined \def \showDOI       #1{#1}\fi
\ifx \showISBNx    \undefined \def \showISBNx     #1{\unskip}     \fi
\ifx \showISBNxiii \undefined \def \showISBNxiii  #1{\unskip}     \fi
\ifx \showISSN     \undefined \def \showISSN      #1{\unskip}     \fi
\ifx \showLCCN     \undefined \def \showLCCN      #1{\unskip}     \fi
\ifx \shownote     \undefined \def \shownote      #1{#1}          \fi
\ifx \showarticletitle \undefined \def \showarticletitle #1{#1}   \fi
\ifx \showURL      \undefined \def \showURL       {\relax}        \fi
\providecommand\bibfield[2]{#2}
\providecommand\bibinfo[2]{#2}
\providecommand\natexlab[1]{#1}
\providecommand\showeprint[2][]{arXiv:#2}

\bibitem[\protect\citeauthoryear{Alt and Godau}{Alt and Godau}{1995}]%
        {alt1995computing}
\bibfield{author}{\bibinfo{person}{Helmut Alt} {and} \bibinfo{person}{Michael
  Godau}.} \bibinfo{year}{1995}\natexlab{}.
\newblock \showarticletitle{Computing the Fr{\'e}chet distance between two
  polygonal curves}.
\newblock \bibinfo{journal}{\emph{International Journal of Computational
  Geometry \& Applications}} \bibinfo{volume}{5}, \bibinfo{number}{01n02}
  (\bibinfo{year}{1995}), \bibinfo{pages}{75--91}.
\newblock


\bibitem[\protect\citeauthoryear{Aung and Tan}{Aung and Tan}{2010}]%
        {aung2010discovery}
\bibfield{author}{\bibinfo{person}{Htoo~Htet Aung} {and}
  \bibinfo{person}{Kian-Lee Tan}.} \bibinfo{year}{2010}\natexlab{}.
\newblock \showarticletitle{Discovery of evolving convoys}. In
  \bibinfo{booktitle}{\emph{International Conference on Scientific and
  Statistical Database Management}}. Springer, \bibinfo{pages}{196--213}.
\newblock


\bibitem[\protect\citeauthoryear{Benkert, Gudmundsson, H{\"u}bner, and
  Wolle}{Benkert et~al\mbox{.}}{2008}]%
        {benkert2008reporting}
\bibfield{author}{\bibinfo{person}{Marc Benkert}, \bibinfo{person}{Joachim
  Gudmundsson}, \bibinfo{person}{Florian H{\"u}bner}, {and}
  \bibinfo{person}{Thomas Wolle}.} \bibinfo{year}{2008}\natexlab{}.
\newblock \showarticletitle{Reporting flock patterns}.
\newblock \bibinfo{journal}{\emph{Computational Geometry}}
  \bibinfo{volume}{41}, \bibinfo{number}{3} (\bibinfo{year}{2008}),
  \bibinfo{pages}{111--125}.
\newblock


\bibitem[\protect\citeauthoryear{Brscic, Kanda, Ikeda, and Miyashita}{Brscic
  et~al\mbox{.}}{2013}]%
        {brscic2013person}
\bibfield{author}{\bibinfo{person}{Drazen Brscic}, \bibinfo{person}{Takayuki
  Kanda}, \bibinfo{person}{Tetsushi Ikeda}, {and} \bibinfo{person}{Takahiro
  Miyashita}.} \bibinfo{year}{2013}\natexlab{}.
\newblock \showarticletitle{Person tracking in large public spaces using 3-D
  range sensors}.
\newblock \bibinfo{journal}{\emph{IEEE Transactions on Human-Machine Systems}}
  \bibinfo{volume}{43}, \bibinfo{number}{6} (\bibinfo{year}{2013}),
  \bibinfo{pages}{522--534}.
\newblock


\bibitem[\protect\citeauthoryear{Eiter and Mannila}{Eiter and Mannila}{1994}]%
        {eiter1994computing}
\bibfield{author}{\bibinfo{person}{Thomas Eiter} {and} \bibinfo{person}{Heikki
  Mannila}.} \bibinfo{year}{1994}\natexlab{}.
\newblock \bibinfo{booktitle}{\emph{Computing discrete Fr{\'e}chet distance}}.
\newblock \bibinfo{type}{{T}echnical {R}eport}.
  \bibinfo{institution}{Citeseer}.
\newblock


\bibitem[\protect\citeauthoryear{Ester, Kriegel, Sander, Xu,
  et~al\mbox{.}}{Ester et~al\mbox{.}}{1996}]%
        {ester1996density}
\bibfield{author}{\bibinfo{person}{Martin Ester}, \bibinfo{person}{Hans-Peter
  Kriegel}, \bibinfo{person}{J{\"o}rg Sander}, \bibinfo{person}{Xiaowei Xu},
  {et~al\mbox{.}}} \bibinfo{year}{1996}\natexlab{}.
\newblock \showarticletitle{A density-based algorithm for discovering clusters
  in large spatial databases with noise.}. In \bibinfo{booktitle}{\emph{Kdd}},
  Vol.~\bibinfo{volume}{96}. \bibinfo{pages}{226--231}.
\newblock


\bibitem[\protect\citeauthoryear{Gaffney and Smyth}{Gaffney and Smyth}{1999}]%
        {gaffney1999trajectory}
\bibfield{author}{\bibinfo{person}{Scott Gaffney} {and}
  \bibinfo{person}{Padhraic Smyth}.} \bibinfo{year}{1999}\natexlab{}.
\newblock \showarticletitle{Trajectory clustering with mixtures of regression
  models}. In \bibinfo{booktitle}{\emph{Proceedings of the fifth ACM SIGKDD
  international conference on Knowledge discovery and data mining}}. ACM,
  \bibinfo{pages}{63--72}.
\newblock


\bibitem[\protect\citeauthoryear{Gupta, Wang, Ari, Hao, Dayal, Mehta, Marwah,
  and Sharma}{Gupta et~al\mbox{.}}{2009}]%
        {gupta2009chaos}
\bibfield{author}{\bibinfo{person}{Chetan Gupta}, \bibinfo{person}{Song Wang},
  \bibinfo{person}{Ismail Ari}, \bibinfo{person}{Ming Hao},
  \bibinfo{person}{Umeshwar Dayal}, \bibinfo{person}{Abhay Mehta},
  \bibinfo{person}{Manish Marwah}, {and} \bibinfo{person}{Ratnesh Sharma}.}
  \bibinfo{year}{2009}\natexlab{}.
\newblock \showarticletitle{Chaos: A data stream analysis architecture for
  enterprise applications}. In \bibinfo{booktitle}{\emph{Commerce and
  Enterprise Computing, 2009. CEC'09. IEEE Conference on}}. IEEE,
  \bibinfo{pages}{33--40}.
\newblock


\bibitem[\protect\citeauthoryear{Huttenlocher, Klanderman, and
  Rucklidge}{Huttenlocher et~al\mbox{.}}{1993}]%
        {huttenlocher1993comparing}
\bibfield{author}{\bibinfo{person}{Daniel~P. Huttenlocher},
  \bibinfo{person}{Gregory~A. Klanderman}, {and} \bibinfo{person}{William~J
  Rucklidge}.} \bibinfo{year}{1993}\natexlab{}.
\newblock \showarticletitle{Comparing images using the Hausdorff distance}.
\newblock \bibinfo{journal}{\emph{IEEE Transactions on pattern analysis and
  machine intelligence}} \bibinfo{volume}{15}, \bibinfo{number}{9}
  (\bibinfo{year}{1993}), \bibinfo{pages}{850--863}.
\newblock


\bibitem[\protect\citeauthoryear{Jensen, Lin, and Ooi}{Jensen
  et~al\mbox{.}}{2007}]%
        {jensen2007continuous}
\bibfield{author}{\bibinfo{person}{Christian~S Jensen}, \bibinfo{person}{Dan
  Lin}, {and} \bibinfo{person}{Beng~Chin Ooi}.}
  \bibinfo{year}{2007}\natexlab{}.
\newblock \showarticletitle{Continuous clustering of moving objects}.
\newblock \bibinfo{journal}{\emph{IEEE Transactions on Knowledge and Data
  Engineering}} \bibinfo{volume}{19}, \bibinfo{number}{9}
  (\bibinfo{year}{2007}).
\newblock


\bibitem[\protect\citeauthoryear{Jeung, Shen, and Zhou}{Jeung
  et~al\mbox{.}}{2008a}]%
        {jeung2008convoy}
\bibfield{author}{\bibinfo{person}{Hoyoung Jeung}, \bibinfo{person}{Heng~Tao
  Shen}, {and} \bibinfo{person}{Xiaofang Zhou}.}
  \bibinfo{year}{2008}\natexlab{a}.
\newblock \showarticletitle{Convoy queries in spatio-temporal databases}. In
  \bibinfo{booktitle}{\emph{Data Engineering, 2008. ICDE 2008. IEEE 24th
  International Conference on}}. IEEE, \bibinfo{pages}{1457--1459}.
\newblock


\bibitem[\protect\citeauthoryear{Jeung, Yiu, Zhou, Jensen, and Shen}{Jeung
  et~al\mbox{.}}{2008b}]%
        {jeung2008discovery}
\bibfield{author}{\bibinfo{person}{Hoyoung Jeung}, \bibinfo{person}{Man~Lung
  Yiu}, \bibinfo{person}{Xiaofang Zhou}, \bibinfo{person}{Christian~S Jensen},
  {and} \bibinfo{person}{Heng~Tao Shen}.} \bibinfo{year}{2008}\natexlab{b}.
\newblock \showarticletitle{Discovery of convoys in trajectory databases}.
\newblock \bibinfo{journal}{\emph{Proceedings of the VLDB Endowment}}
  \bibinfo{volume}{1}, \bibinfo{number}{1} (\bibinfo{year}{2008}),
  \bibinfo{pages}{1068--1080}.
\newblock


\bibitem[\protect\citeauthoryear{Jin, Hua, Xu, and Zhou}{Jin
  et~al\mbox{.}}{2019}]%
        {Jin2019MovingOL}
\bibfield{author}{\bibinfo{person}{Fengmei Jin}, \bibinfo{person}{Wen Hua},
  \bibinfo{person}{Jiajie Xu}, {and} \bibinfo{person}{Xiaofang Zhou}.}
  \bibinfo{year}{2019}\natexlab{}.
\newblock \showarticletitle{Moving Object Linking Based on Historical Trace}.
\newblock \bibinfo{journal}{\emph{2019 IEEE 35th International Conference on
  Data Engineering (ICDE)}} (\bibinfo{year}{2019}),
  \bibinfo{pages}{1058--1069}.
\newblock


\bibitem[\protect\citeauthoryear{Kalnis, Mamoulis, and Bakiras}{Kalnis
  et~al\mbox{.}}{2005}]%
        {kalnis2005discovering}
\bibfield{author}{\bibinfo{person}{Panos Kalnis}, \bibinfo{person}{Nikos
  Mamoulis}, {and} \bibinfo{person}{Spiridon Bakiras}.}
  \bibinfo{year}{2005}\natexlab{}.
\newblock \showarticletitle{On discovering moving clusters in spatio-temporal
  data}. In \bibinfo{booktitle}{\emph{International Symposium on Spatial and
  Temporal Databases}}. Springer, \bibinfo{pages}{364--381}.
\newblock


\bibitem[\protect\citeauthoryear{Lan, Yu, Cao, Song, and Wang}{Lan
  et~al\mbox{.}}{2017}]%
        {lan2017discovering}
\bibfield{author}{\bibinfo{person}{Ruoshan Lan}, \bibinfo{person}{Yanwei Yu},
  \bibinfo{person}{Lei Cao}, \bibinfo{person}{Peng Song}, {and}
  \bibinfo{person}{Yingjie Wang}.} \bibinfo{year}{2017}\natexlab{}.
\newblock \showarticletitle{Discovering Evolving Moving Object Groups from
  Massive-Scale Trajectory Streams}. In \bibinfo{booktitle}{\emph{Mobile Data
  Management (MDM), 2017 18th IEEE International Conference on}}. IEEE,
  \bibinfo{pages}{256--265}.
\newblock


\bibitem[\protect\citeauthoryear{Laube and Imfeld}{Laube and Imfeld}{2002}]%
        {laube2002analyzing}
\bibfield{author}{\bibinfo{person}{Patrick Laube} {and}
  \bibinfo{person}{Stephan Imfeld}.} \bibinfo{year}{2002}\natexlab{}.
\newblock \showarticletitle{Analyzing relative motion within groups oftrackable
  moving point objects}. In \bibinfo{booktitle}{\emph{International Conference
  on Geographic Information Science}}. Springer, \bibinfo{pages}{132--144}.
\newblock


\bibitem[\protect\citeauthoryear{Lee, Han, Li, and Gonzalez}{Lee
  et~al\mbox{.}}{2008}]%
        {lee2008traclass}
\bibfield{author}{\bibinfo{person}{Jae-Gil Lee}, \bibinfo{person}{Jiawei Han},
  \bibinfo{person}{Xiaolei Li}, {and} \bibinfo{person}{Hector Gonzalez}.}
  \bibinfo{year}{2008}\natexlab{}.
\newblock \showarticletitle{TraClass: trajectory classification using
  hierarchical region-based and trajectory-based clustering}.
\newblock \bibinfo{journal}{\emph{Proceedings of the VLDB Endowment}}
  \bibinfo{volume}{1}, \bibinfo{number}{1} (\bibinfo{year}{2008}),
  \bibinfo{pages}{1081--1094}.
\newblock


\bibitem[\protect\citeauthoryear{Lee, Han, and Whang}{Lee
  et~al\mbox{.}}{2007}]%
        {lee2007trajectory}
\bibfield{author}{\bibinfo{person}{Jae-Gil Lee}, \bibinfo{person}{Jiawei Han},
  {and} \bibinfo{person}{Kyu-Young Whang}.} \bibinfo{year}{2007}\natexlab{}.
\newblock \showarticletitle{Trajectory clustering: a partition-and-group
  framework}. In \bibinfo{booktitle}{\emph{Proceedings of the 2007 ACM SIGMOD
  international conference on Management of data}}. ACM,
  \bibinfo{pages}{593--604}.
\newblock


\bibitem[\protect\citeauthoryear{Li, Ceikute, Jensen, and Tan}{Li
  et~al\mbox{.}}{2013}]%
        {li2013effective}
\bibfield{author}{\bibinfo{person}{Xiaohui Li}, \bibinfo{person}{Vaida
  Ceikute}, \bibinfo{person}{Christian~S Jensen}, {and}
  \bibinfo{person}{Kian-Lee Tan}.} \bibinfo{year}{2013}\natexlab{}.
\newblock \showarticletitle{Effective online group discovery in trajectory
  databases}.
\newblock \bibinfo{journal}{\emph{IEEE Transactions on Knowledge and Data
  Engineering}} \bibinfo{volume}{25}, \bibinfo{number}{12}
  (\bibinfo{year}{2013}), \bibinfo{pages}{2752--2766}.
\newblock


\bibitem[\protect\citeauthoryear{Li, Ding, Han, and Kays}{Li
  et~al\mbox{.}}{2010a}]%
        {li2010swarm}
\bibfield{author}{\bibinfo{person}{Zhenhui Li}, \bibinfo{person}{Bolin Ding},
  \bibinfo{person}{Jiawei Han}, {and} \bibinfo{person}{Roland Kays}.}
  \bibinfo{year}{2010}\natexlab{a}.
\newblock \showarticletitle{Swarm: Mining relaxed temporal moving object
  clusters}.
\newblock \bibinfo{journal}{\emph{Proceedings of the VLDB Endowment}}
  \bibinfo{volume}{3}, \bibinfo{number}{1-2} (\bibinfo{year}{2010}),
  \bibinfo{pages}{723--734}.
\newblock


\bibitem[\protect\citeauthoryear{Li, Han, Ji, Tang, Yu, Ding, Lee, and Kays}{Li
  et~al\mbox{.}}{2011}]%
        {li2011movemine}
\bibfield{author}{\bibinfo{person}{Zhenhui Li}, \bibinfo{person}{Jiawei Han},
  \bibinfo{person}{Ming Ji}, \bibinfo{person}{Lu-An Tang},
  \bibinfo{person}{Yintao Yu}, \bibinfo{person}{Bolin Ding},
  \bibinfo{person}{Jae-Gil Lee}, {and} \bibinfo{person}{Roland Kays}.}
  \bibinfo{year}{2011}\natexlab{}.
\newblock \showarticletitle{Movemine: Mining moving object data for discovery
  of animal movement patterns}.
\newblock \bibinfo{journal}{\emph{ACM Transactions on Intelligent Systems and
  Technology (TIST)}} \bibinfo{volume}{2}, \bibinfo{number}{4}
  (\bibinfo{year}{2011}), \bibinfo{pages}{37}.
\newblock


\bibitem[\protect\citeauthoryear{Li, Lee, Li, and Han}{Li
  et~al\mbox{.}}{2010b}]%
        {li2010incremental}
\bibfield{author}{\bibinfo{person}{Zhenhui Li}, \bibinfo{person}{Jae-Gil Lee},
  \bibinfo{person}{Xiaolei Li}, {and} \bibinfo{person}{Jiawei Han}.}
  \bibinfo{year}{2010}\natexlab{b}.
\newblock \showarticletitle{Incremental clustering for trajectories}. In
  \bibinfo{booktitle}{\emph{International Conference on Database Systems for
  Advanced Applications}}. Springer, \bibinfo{pages}{32--46}.
\newblock


\bibitem[\protect\citeauthoryear{Naserian, Wang, Xu, and Dong}{Naserian
  et~al\mbox{.}}{2016}]%
        {naserian2016discovery}
\bibfield{author}{\bibinfo{person}{Elahe Naserian}, \bibinfo{person}{Xinheng
  Wang}, \bibinfo{person}{Xiaolong Xu}, {and} \bibinfo{person}{Yuning Dong}.}
  \bibinfo{year}{2016}\natexlab{}.
\newblock \showarticletitle{Discovery of Loose Travelling Companion Patterns
  from Human Trajectories}. In \bibinfo{booktitle}{\emph{High Performance
  Computing and Communications; IEEE 14th International Conference on Smart
  City; IEEE 2nd International Conference on Data Science and Systems
  (HPCC/SmartCity/DSS), 2016 IEEE 18th International Conference on}}. IEEE,
  \bibinfo{pages}{1238--1245}.
\newblock


\bibitem[\protect\citeauthoryear{Naserian, Wang, Xu, and Dong}{Naserian
  et~al\mbox{.}}{2018}]%
        {naserian2018framework}
\bibfield{author}{\bibinfo{person}{Elahe Naserian}, \bibinfo{person}{Xinheng
  Wang}, \bibinfo{person}{Xiaolong Xu}, {and} \bibinfo{person}{Yuning Dong}.}
  \bibinfo{year}{2018}\natexlab{}.
\newblock \showarticletitle{A framework of loose travelling companion discovery
  from human trajectories}.
\newblock \bibinfo{journal}{\emph{IEEE Transactions on Mobile Computing}}
  \bibinfo{volume}{17}, \bibinfo{number}{11} (\bibinfo{year}{2018}),
  \bibinfo{pages}{2497--2511}.
\newblock


\bibitem[\protect\citeauthoryear{Tang, Zheng, Yuan, Han, Leung, Hung, and
  Peng}{Tang et~al\mbox{.}}{2012}]%
        {tang2012discovery}
\bibfield{author}{\bibinfo{person}{Lu-An Tang}, \bibinfo{person}{Yu Zheng},
  \bibinfo{person}{Jing Yuan}, \bibinfo{person}{Jiawei Han},
  \bibinfo{person}{Alice Leung}, \bibinfo{person}{Chih-Chieh Hung}, {and}
  \bibinfo{person}{Wen-Chih Peng}.} \bibinfo{year}{2012}\natexlab{}.
\newblock \showarticletitle{On discovery of traveling companions from streaming
  trajectories}. In \bibinfo{booktitle}{\emph{Data Engineering (ICDE), 2012
  IEEE 28th International Conference on}}. IEEE, \bibinfo{pages}{186--197}.
\newblock


\bibitem[\protect\citeauthoryear{Toohey and Duckham}{Toohey and
  Duckham}{2015}]%
        {toohey2015trajectory}
\bibfield{author}{\bibinfo{person}{Kevin Toohey} {and} \bibinfo{person}{Matt
  Duckham}.} \bibinfo{year}{2015}\natexlab{}.
\newblock \showarticletitle{Trajectory similarity measures}.
\newblock \bibinfo{journal}{\emph{Sigspatial Special}} \bibinfo{volume}{7},
  \bibinfo{number}{1} (\bibinfo{year}{2015}), \bibinfo{pages}{43--50}.
\newblock


\bibitem[\protect\citeauthoryear{Vieira, Bakalov, and Tsotras}{Vieira
  et~al\mbox{.}}{2009}]%
        {vieira2009line}
\bibfield{author}{\bibinfo{person}{Marcos~R Vieira}, \bibinfo{person}{Petko
  Bakalov}, {and} \bibinfo{person}{Vassilis~J Tsotras}.}
  \bibinfo{year}{2009}\natexlab{}.
\newblock \showarticletitle{On-line discovery of flock patterns in
  spatio-temporal data}. In \bibinfo{booktitle}{\emph{Proceedings of the 17th
  ACM SIGSPATIAL international conference on advances in geographic information
  systems}}. ACM, \bibinfo{pages}{286--295}.
\newblock


\bibitem[\protect\citeauthoryear{Wang, Lim, and Hwang}{Wang
  et~al\mbox{.}}{2006}]%
        {wang2006efficient}
\bibfield{author}{\bibinfo{person}{Yida Wang}, \bibinfo{person}{Ee-Peng Lim},
  {and} \bibinfo{person}{San-Yih Hwang}.} \bibinfo{year}{2006}\natexlab{}.
\newblock \showarticletitle{Efficient mining of group patterns from user
  movement data}.
\newblock \bibinfo{journal}{\emph{Data \& Knowledge Engineering}}
  \bibinfo{volume}{57}, \bibinfo{number}{3} (\bibinfo{year}{2006}),
  \bibinfo{pages}{240--282}.
\newblock


\bibitem[\protect\citeauthoryear{Wang, Luo, Xiong, Prosser, Newman, Takekawa,
  and Yan}{Wang et~al\mbox{.}}{2015}]%
        {wang2015discovering}
\bibfield{author}{\bibinfo{person}{Yuwei Wang}, \bibinfo{person}{Ze Luo},
  \bibinfo{person}{Yan Xiong}, \bibinfo{person}{Diann~J Prosser},
  \bibinfo{person}{Scott~H Newman}, \bibinfo{person}{John~Y Takekawa}, {and}
  \bibinfo{person}{Baoping Yan}.} \bibinfo{year}{2015}\natexlab{}.
\newblock \showarticletitle{Discovering loose group movement patterns from
  animal trajectories}. In \bibinfo{booktitle}{\emph{e-Science (e-Science),
  2015 IEEE 11th International Conference on}}. IEEE,
  \bibinfo{pages}{196--206}.
\newblock


\bibitem[\protect\citeauthoryear{Xie, Li, and Phillips}{Xie
  et~al\mbox{.}}{2017}]%
        {xie2017distributed}
\bibfield{author}{\bibinfo{person}{Dong Xie}, \bibinfo{person}{Feifei Li},
  {and} \bibinfo{person}{Jeff~M Phillips}.} \bibinfo{year}{2017}\natexlab{}.
\newblock \showarticletitle{Distributed trajectory similarity search}.
\newblock \bibinfo{journal}{\emph{Proceedings of the VLDB Endowment}}
  \bibinfo{volume}{10}, \bibinfo{number}{11} (\bibinfo{year}{2017}),
  \bibinfo{pages}{1478--1489}.
\newblock


\bibitem[\protect\citeauthoryear{Yu, Cao, Rundensteiner, and Wang}{Yu
  et~al\mbox{.}}{2014}]%
        {yu2014detecting}
\bibfield{author}{\bibinfo{person}{Yanwei Yu}, \bibinfo{person}{Lei Cao},
  \bibinfo{person}{Elke~A Rundensteiner}, {and} \bibinfo{person}{Qin Wang}.}
  \bibinfo{year}{2014}\natexlab{}.
\newblock \showarticletitle{Detecting moving object outliers in massive-scale
  trajectory streams}. In \bibinfo{booktitle}{\emph{Proceedings of the 20th ACM
  SIGKDD international conference on Knowledge discovery and data mining}}.
  ACM, \bibinfo{pages}{422--431}.
\newblock


\bibitem[\protect\citeauthoryear{Yu, Wang, Wang, Wang, and He}{Yu
  et~al\mbox{.}}{2013}]%
        {yu2013online}
\bibfield{author}{\bibinfo{person}{Yanwei Yu}, \bibinfo{person}{Qin Wang},
  \bibinfo{person}{Xiaodong Wang}, \bibinfo{person}{Huan Wang}, {and}
  \bibinfo{person}{Jie He}.} \bibinfo{year}{2013}\natexlab{}.
\newblock \showarticletitle{Online clustering for trajectory data stream of
  moving objects}.
\newblock \bibinfo{journal}{\emph{Computer science and information systems}}
  \bibinfo{volume}{10}, \bibinfo{number}{3} (\bibinfo{year}{2013}),
  \bibinfo{pages}{1293--1317}.
\newblock


\bibitem[\protect\citeauthoryear{Yuan, Zheng, Xie, and Sun}{Yuan
  et~al\mbox{.}}{2013}]%
        {yuan2013t}
\bibfield{author}{\bibinfo{person}{Jing Yuan}, \bibinfo{person}{Yu Zheng},
  \bibinfo{person}{Xing Xie}, {and} \bibinfo{person}{Guangzhong Sun}.}
  \bibinfo{year}{2013}\natexlab{}.
\newblock \showarticletitle{T-drive: Enhancing driving directions with taxi
  drivers' intelligence}.
\newblock \bibinfo{journal}{\emph{IEEE Transactions on Knowledge and Data
  Engineering}} \bibinfo{volume}{25}, \bibinfo{number}{1}
  (\bibinfo{year}{2013}), \bibinfo{pages}{220--232}.
\newblock


\bibitem[\protect\citeauthoryear{Zanlungo, Br{\v{s}}{\v{c}}i{\'c}, and
  Kanda}{Zanlungo et~al\mbox{.}}{2015}]%
        {zanlungo2015spatial}
\bibfield{author}{\bibinfo{person}{Francesco Zanlungo},
  \bibinfo{person}{Dra{\v{z}}en Br{\v{s}}{\v{c}}i{\'c}}, {and}
  \bibinfo{person}{Takayuki Kanda}.} \bibinfo{year}{2015}\natexlab{}.
\newblock \showarticletitle{Spatial-size scaling of pedestrian groups under
  growing density conditions}.
\newblock \bibinfo{journal}{\emph{Physical Review E}} \bibinfo{volume}{91},
  \bibinfo{number}{6} (\bibinfo{year}{2015}), \bibinfo{pages}{062810}.
\newblock


\bibitem[\protect\citeauthoryear{Zanlungo, Ikeda, and Kanda}{Zanlungo
  et~al\mbox{.}}{2014}]%
        {zanlungo2014potential}
\bibfield{author}{\bibinfo{person}{Francesco Zanlungo},
  \bibinfo{person}{Tetsushi Ikeda}, {and} \bibinfo{person}{Takayuki Kanda}.}
  \bibinfo{year}{2014}\natexlab{}.
\newblock \showarticletitle{Potential for the dynamics of pedestrians in a
  socially interacting group}.
\newblock \bibinfo{journal}{\emph{Physical Review E}} \bibinfo{volume}{89},
  \bibinfo{number}{1} (\bibinfo{year}{2014}), \bibinfo{pages}{012811}.
\newblock


\bibitem[\protect\citeauthoryear{Zheng, Zheng, Yuan, and Shang}{Zheng
  et~al\mbox{.}}{2013}]%
        {zheng2013discovery}
\bibfield{author}{\bibinfo{person}{Kai Zheng}, \bibinfo{person}{Yu Zheng},
  \bibinfo{person}{Nicholas~Jing Yuan}, {and} \bibinfo{person}{Shuo Shang}.}
  \bibinfo{year}{2013}\natexlab{}.
\newblock \showarticletitle{On discovery of gathering patterns from
  trajectories}. In \bibinfo{booktitle}{\emph{Data Engineering (ICDE), 2013
  IEEE 29th International Conference on}}. IEEE, \bibinfo{pages}{242--253}.
\newblock


\bibitem[\protect\citeauthoryear{Zheng, Zheng, Yuan, Shang, and Zhou}{Zheng
  et~al\mbox{.}}{2014}]%
        {zheng2014online}
\bibfield{author}{\bibinfo{person}{Kai Zheng}, \bibinfo{person}{Yu Zheng},
  \bibinfo{person}{Nicholas~J Yuan}, \bibinfo{person}{Shuo Shang}, {and}
  \bibinfo{person}{Xiaofang Zhou}.} \bibinfo{year}{2014}\natexlab{}.
\newblock \showarticletitle{Online discovery of gathering patterns over
  trajectories}.
\newblock \bibinfo{journal}{\emph{IEEE Transactions on Knowledge and Data
  Engineering}} \bibinfo{volume}{26}, \bibinfo{number}{8}
  (\bibinfo{year}{2014}), \bibinfo{pages}{1974--1988}.
\newblock


\bibitem[\protect\citeauthoryear{Zheng}{Zheng}{2015}]%
        {zheng2015trajectory}
\bibfield{author}{\bibinfo{person}{Yu Zheng}.} \bibinfo{year}{2015}\natexlab{}.
\newblock \showarticletitle{Trajectory data mining: an overview}.
\newblock \bibinfo{journal}{\emph{ACM Transactions on Intelligent Systems and
  Technology (TIST)}} \bibinfo{volume}{6}, \bibinfo{number}{3}
  (\bibinfo{year}{2015}), \bibinfo{pages}{29}.
\newblock


\end{thebibliography}

\end{document}